\documentclass[10pt, oneside, reqno]{amsart}
\usepackage{graphicx}

\usepackage{indentfirst,csquotes}
\usepackage[a4paper, total={6in, 8in}]{geometry}
\usepackage{footnote}
\usepackage{hyperref}
\usepackage{amssymb,amsthm,amsmath, mathrsfs, amsfonts}
\usepackage{enumerate,lineno}
\usepackage{url}
\usepackage{xcolor, paralist,fancyhdr,etoolbox}
\usepackage{graphicx}
\usepackage{tikz}
\usepackage[normalem]{ulem}
\usepackage{float}
\usepackage{caption,subcaption}
\usepackage{multicol,multirow}
\usepackage{bookmark}
\usepackage{relsize}
\usepackage[dvipsnames]{xcolor}

\theoremstyle{plain}
\newtheorem{theo}{Theorem}[section]
\newtheorem{cor}{Corollary}[section]
\newtheorem{rem}{Remark}[section]
\newtheorem{defi}{Definition}[section]
\newtheorem{lemma}{Lemma}[section]
\newtheorem{prop}{Proposition}[section]
\newtheorem{ex}{Example}[section]

\newtheorem{op}{Open Problem}

\def\Tf{{\mathcal F}_{2,n}}

\newcommand{\F}{\mathbb{F}}

\newcommand{\UB}{\mathcal{UB}}
\newcommand{\Ftn}{\mathbb{F}_{2^n}}
\newcommand{\Ftm}{\mathbb{F}_{2^m}}

\newcommand{\de}{\delta}

\let\svthefootnote\thefootnote

\makeatletter
\newcommand{\addresseshere}{%
  \enddoc@text\let\enddoc@text\relax
 }
 \makeatother

\makeatletter
\renewcommand\subsection{\@startsection{subsection}{2}%
  \z@{-.5\linespacing\@plus-.7\linespacing}{.5\linespacing}%
  {\normalfont\bfseries}}
\makeatother

\makeatletter
\renewcommand\subsubsection{\@startsection{subsubsection}{2}%
  \z@{-.5\linespacing\@plus-.7\linespacing}{.5\linespacing}%
  {\normalfont\itshape}} 
\makeatother

\usepackage{scalerel}
\newcommand{\asterisk}{\scaleobj{0.65}{*}}
\newcommand{\dagg}{\scaleobj{0.65}{\dag}}

\begin{document}
\title{On different notions related to APN mappings}
\author[Rodríguez et al.]{René Rodríguez-Aldama \and Ajla Šehović \and Enes Pasalic \and Sadmir Kudin}
\date{\today}
\address{Univerza na Primorskem - Fakulteta za matematiko, naravoslovje in informacijske tehnologije. Glagoljaška ulica 8, 6000, Koper, Slovenia.}
%Inštitut Andrej Marušič. Muzejski trg 2, 6000, Koper, Slovenia.}
\email{X@iam.upr.si, Y@gmail.com; $X\in\{ ajla.sehovic, sadmir.kudin \}$, $Y\in \{enes.pasalic6, rene7ca\}$}

\let\thefootnote\relax
\footnotetext{MSC2020: Primary 94A60, Secondary 11T06.}

\addtocounter{footnote}{-1}\let\thefootnote\svthefootnote

\setlength{\intextsep}{2pt}

\begin{abstract}

An \emph{almost perfect nonlinear} (APN) mapping $F:\Ftn\to \Ftn$ is a polynomial characterized by the {non-vanishing property} on $2$-dimensional affine subspaces ($2$-flats), that is, the summation of $F$ on any $2$-flat is not null. In this work, we analyze notions that are closely related to the non-vanishing property of APN polynomials. First, aiming to understand which $k$-dimensional flats of $\F_{2^n}$ remain flats under $F$, we study the \emph{$k$-breaking property}. In the context of cryptography, Hou initiated this line of investigation for APN permutations in \cite{Hou2006}. Later, Kolomeec and Bykov~\cite{KolomeecDCC2024} studied the $k$-breaking property of the multiplicative inverse permutation. We extend these studies to general mappings and characterize the $2$-breaking property of APN functions. As a byproduct of our analysis, we present a lower bound on the vectorial nonlinearity in terms of the breaking property. 
Recently, Carlet \cite{CarletTwoNotions2024} has introduced two generalizations of the APN property to any vectorial mapping: \emph{$k$-strongly non-normality} and \emph{$k$-th-order sum-freedom}, where a function is $k$-strongly non-normal if its restriction to $k$-flats is not affine. Sum-freedom generalizes the non-vanishing property of APN functions to higher dimensional flats. We provide in-depth observations of the relations between the breaking property, strongly non-normality and sum-freedom. To get more insight into these intertwined notions, we introduce a fourth concept called \emph{$k$-strongly breaking}, which implies the standard breaking property. We derive several structural results for both notions. Moreover, a characterization of a subclass of APN functions in terms of the $2$-strongly breaking property is given. Finally, we propose a different perspective of the non-vanishing property via a natural character transformation first suggested by Carlet \cite{CarletBook}.  On the one hand, this transformation is closely related to the sum-of-square indicator of the components of a function. In particular, we derive a precise value for the total sum of the sum-of-square indicators of a polynomial $F$. Remarkably, by employing this approach, we provide a simple answer to Open Problem 4 in \cite{BeCaChLa06}. On the other hand, this technique allows us to explore balancedness properties related to polynomials. One of these properties characterizes \emph{component-wise APNness}, introduced in \cite{CarletCAPN18}, for odd dimensions, and provides a natural extension to any dimension regardless of its parity. We show that notable APN functions such as Dillon's APN permutation and the Gold functions satisfy a connected property, termed \emph{$k$-balanced}, which is presented under our framework.

\end{abstract}

\keywords{APN functions,  balancedness, affine subspaces, vanishing flats, polynomials.}

\maketitle

\section{Introduction}

In modern cryptography, functions used in cryptographic systems must provide strong resistance against various types of attacks. One of the most important threats is \emph{differential cryptanalysis}, initiated by Biham and Shamir \cite{BiSha91}, which exploits how differences in the inputs of a (vectorial) $(n,m)$-function $F:\Ftn \to \F_{2^m}$ affect the corresponding differences in its outputs. To measure the resistance of a mapping $F$ against differential attacks, Nyberg introduced the notion of \emph{differential uniformity} \cite{Nyberg94}. Polynomials, that is, $(n,n)$-functions, with differential uniformity equal to $2$, which provide optimal resistance to differential cryptanalysis, are referred to as \emph{almost perfect nonlinear (APN)} functions. Consequently, they have become a crucial object of study in finite field theory and cryptography.  APN functions are also of significant interest in coding theory, as they give rise to linear codes of length $2^n-1$ and dimension $2^{n}-1-2n$ with  minimum distance 5 \cite{CCZ98}. Moreover,  APN permutations are closely related to several areas of discrete mathematics, such as  difference sets~\cite{Dillon1999,DiDo04}, Sidon sets \cite{CaPi23,CzPo26} and incidence structures in finite geometry \cite{Yoshiara10}. In the context of cryptography, Beth and Ding were among the first to investigate APN permutations and their suitability as nonlinear components in block ciphers \cite{BeDi94}, identifying the existence of such permutations in even dimensions as an open problem. It is known that there are no APN permutations for $n=4$ \cite{CaSaVi17}, whereas an example was found for $n=6$ \cite{Dillon2010}, now known as Dillon's permutation. Whether APN permutations exist in even dimensions $n\geq 8$ remains an open problem to this day, commonly referred to as \emph{the big APN problem}.
The first and most extensively studied APN functions are power functions $F(x)=x^d$ over $\mathbb{F}_{2^n}$. Six infinite families of APN monomials are known, including the inverse function $x\mapsto x^{2^n-2}$ \cite{BeDi94,Nyberg94}, Gold function $x\mapsto x^{2^k+1}$ \cite{Gold68,Nyberg94}, the Kasami function $x\mapsto x^{2^{2k}-2^k+1}$ \cite{JanWill93,Kasami71}, where $\gcd(n,k)=1$, as well as the Welch \cite{Welch74}, Niho \cite{Niho72} and Dobbertin functions \cite{Dobbertin03}.  The construction of APN functions outside the classical power families has become a major research direction. Many infinite families of APN functions have been constructed over the years (see \cite{GoKo25,LiKa24} for an overview). 

The subsequent search for APN functions outside the known power families has therefore been accompanied by the problem of determining when two functions should be regarded as essentially different.  This problem naturally led to the study of (several) equivalence relations, among which are \emph{linear (resp. affine)} equivalence, \emph{extended-affine (EA)} equivalence and, more generally, \emph{Carlet–Charpin–Zinoviev (CCZ)} equivalence, introduced in \cite{CCZ98}.  CCZ-equivalence preserves differential uniformity and hence the APN property, which has motivated extensive classification efforts. Brinkmann and Leander completely classified APN functions in dimensions 4 and 5 with respect to  CCZ-equivalence and showed that every APN function in dimension up to 5 is CCZ-equivalent to a power function \cite{BrLe08}. For $n\geq 6$, the classification of APN functions requires additional restrictions, such as imposing a bound on their algebraic degree.
 In dimension six, Dillon's computational work produced a list of thirteen CCZ-inequivalent quadratic APN functions, which was shown to be complete by Edel \cite{Dillon2006,Edel2011}. Additionally, all cubic APN functions in dimension $6$ have been classified \cite{Langevin12},  namely, there is only one such class with Edel-Pott function as its representative \cite{EdelPott2009}. Further investigations by Edel and Pott, Kalgin and Idrisova, and other authors have considerably extended the classification and enumeration of quadratic APN functions in dimension $n=7$, ultimately having obtained a complete classification of quadratic APN functions in seven variables, consisting of $488$ CCZ-inequivalent classes \cite{EdelPott2009,KalId23, YuWangLi2014}.  These computations also illustrate how rapidly the number of inequivalent APN functions grows. In particular, recent computational work of Beierle et al. produced $3,775,599$ inequivalent quadratic APN functions in dimension $n=8$ and estimated the total number of such functions to be of the order of six million \cite{BeLaLePoRa25}. Thus, beyond the smallest dimensions, the classification problem is no longer simply a matter of listing all known constructions, but requires structural invariants and efficient computational methods for distinguishing equivalence classes.

A fundamental study of the structural properties of APN functions was given by Berger et al. \cite{BeCaChLa06}, who established several characterizations of APN functions and APN permutations in terms of their component functions. In particular, a function $F$ is APN if and only if the sum of its \emph{sum-of-square indicator} of its nonzero components $x\mapsto \mathrm{Tr}_n(vF(x))$, $v\in \Ftn^*$, is equal to $2^{2n+1}(2^n-1)$. Further structural restrictions have also been obtained from the component functions and from the behavior of the image of an APN mapping. It has been shown that an APN permutation cannot have a quadratic component \cite{CaSaVi17}, and lower and upper bounds on the size of the image of an APN function have been established in \cite{GoharDCC2023}. Particularly, for a non-bijective APN function $F$, it holds $\frac{2^n+1}{3}+1\leq\operatorname{Im}(F)\leq 2^n-2^{\frac{n-1}{2}} $. APN functions have been studied through many lenses, from  ortho-derivatives  \cite{CaCoPe24} to  topological methods applied to the associated linear codes \cite{ZhLiZh25}. Vitkup studied symmetry properties and the range of APN functions, obtaining bounds on the number of symmetric coordinate functions and on the multiplicities of values in the image of an APN function \cite{Vitkup2016}. More recently, new parameters derived from the behavior of second-order derivatives have been introduced for Boolean and vectorial Boolean functions. These parameters are invariant under EA equivalence and they allow studying APN functions of degrees two and three \cite{Villa2024}. 

The classical notion of APNness has inspired several generalizations. One of them is \emph{almost perfect $c$-nonlinearity} (related to AP$c$N functions), introduced  in \cite{ElFeRiStTk2020}. For $c\in\mathbb F_{2^n}$, one considers the equation $F(x)+cF(x+a)=b,$ and the $c$-\emph{differential uniformity} is the maximum number of its solutions over $(a,b)\in\Ftn\times \Ftn$ (with $a\neq0$ when $c=1$). The case $c=1$ gives the usual differential uniformity, so that APcN functions are precisely the APN functions in this case. An $(n,m)$-function $F$ is said to be \emph{weakly APN} \cite{ArCalMaSa16} if, for every nonzero $a\in\Ftn$, the derivative $D_aF$ has an image of cardinality strictly greater than $2^{n-2}$. It is called \emph{partially APN} \cite{BuKaKwRiSt20} if there exists some $c\in\Ftn$ such that, for every affine plane containing $c$, the sum of the values of $F$ over that plane is nonzero. More recently, Carlet proposed two higher-order generalizations of APNness, namely \emph{strong non-normality} and \emph{sum-freedom}. A function $F$ is said to be $k$-strongly non-normal if, for every $k$-dimensional affine subspace $A\subseteq\Ftn$ (or $k$-flat), the restriction $F|_A$ is non-affine, and $F$ is said to be $k$th-order sum-free if the sum of its values over every $k$-flat is nonzero \cite{CarletTwoNotions2024}. The $2$nd-order sum-freedom recovers APNness and it is known as the \emph{non-vanishing property} on $2$-flats. Moreover, there is a connection between 3rd-order sum-freedom and the so-called {\em D-property} \cite{TaniguchiD-prop} of an $(n,m)$-function $F$, saying that $\{F(x)+F(y)+F(z) + F(x+y+z) : x,y,z \in \Ftn \} = \F_{2^m}$. The conjecture regarding the $k$th-order sum-freedom of the inverse function has been heavily studied in \cite{CarletHou2025, EbHoRyZh26}. Recently the conjecture has been confirmed for non-prime dimensions \cite{HouZhao2026}. In general, for $k\geq 3$, $k$th-order sum-free functions seem to be rare since there are just a few known examples \cite{NonVanishing}.
 
 The study of APN functions has then evolved from a cryptographic problem concerning optimal resistance to differential attacks into a broad mathematical theory involving finite fields, coding theory, finite geometry, combinatorics, and computational classification. The main goal of this work is to study concepts that are closely related to the APN property. Primarily, we will be interested in dealing with notions in the nearest orbit of the non-vanishing property.

One of the most natural ways to analyze the behaviour of mappings is by exploring invariant subsets, e.g., (invariant subspaces, fixed subfields, stabilizer of a group action). The invariant property relevant to our context is the property of being a flat: we are interested in determining when the image of a $k$-flat under a mapping $F$ is no longer a flat. These flats are termed broken flats. In the context of cryptographically significant functions, a study of the breaking properties of permutations (including APN permutations) was carried out in \cite{Hou2006} under the term \emph{affinity}. This study was later refined for general permutations in \cite{ClarkHouMihailovs2007, Kolomeec2024A, Kolomeec2024}. Some observations for general differentially uniform functions were given in \cite{Kolomeec2023}. The study of unbroken $k$-flats for the multiplicative inverse was considered recently in \cite{KolomeecDCC2024}. In particular, it was shown that the function $F(x)=x^{2^n-2}$ breaks all affine subspaces except for multiplicative cosets of subfields. In addition, a sufficient condition was provided so that a function $A(F(x)) + b$ has no invariant affine subspaces $U$ of cardinality $2 < |U| < 2^n$ for an invertible linear transformation $A : \F_{2^n} \rightarrow \F_{2^n}$ and $b \in  \F_{2^n}$. An additional motivation to study the breaking concept for APN functions comes from the Kim mapping $kim(x):=x^3+x^{10}+\omega x^{24}$ over $\mathbb{F}_{2^6}$, where $\omega^6+\omega^4+\omega^3+\omega+1=0$. Namely, it was shown \cite{Dillon2010} that the image of a multiplicative coset of $\mathbb{F}_{2^3}$ under $kim(x)$ is again a multiplicative coset. Two related works studying a generalized version of the unbroken property are given in \cite{BraWoPre05,Enrico}.  It turns out that the breaking property of all $k$-flats is invariant under affine equivalence. The precise number of unbroken $2$-flats can be derived for a subclass of functions which includes all APN functions (see Theorem \ref{th:unbroken_general}). Moreover, a characterization of the APN property can be derived using the $2$-breaking property. Additionally, we provide a more detailed description of the unbroken flats of polynomials whose preimages are upper bounded by a small number.

Cryptographic functions used in the design of some stream ciphers, such as the nonlinear combiner model and the nonlinear filter model,  must be far from  all affine functions in order to render best affine approximation attacks inefficient (see \cite{Ding1991}). Then, for an $(n, m)$-function $F$, one studies the behavior of functions $F+L$, where $L$ is any affine $(n, m)$-function. This behaviour is captured in the notion of vectorial nonlinearity, introduced in \cite{VectorialNonlinearity} and later studied in \cite{LiuChenMesn2017}. We present a lower bound on the vectorial nonlinearity of a polynomial in terms of the breaking property of polynomials, improving on a lower bound proposed in \cite{LiuChenMesn2017}. However, it is worth mentioning that a better bound has been given in \cite{Nagy2025} using the properties of Sidon sets, which is by far the best known lower bound. 

Since the concepts of strongly non-normality and sum-freedom are generalizations of the non-vanishing property, it is natural to compare them with the breaking property of mappings. It turns out that the breaking property is a strengthening of strongly non-normality (Proposition \ref{prop:break-snn}). On the other hand, the breaking property and sum-freedom are different concepts in general. Moreover, when treating $2$-dimensional flats, these three concepts are almost equivalent (see Figure \ref{fig:rel2}). For APN functions, $k$th-order sum-freedom implies $k$-breaking (Proposition \ref{prop:sumfree_breaking}). As a matter of fact, the non-vanishing property hides some structural aspects that are worth investigating, thus we introduce the concept of \emph{strongly breaking} property aiming to address this issue and to better understand the connections between these notions (see Definition \ref{def:sbreaking}). As its name suggests, the strongly breaking property implies the standard breaking property in all non-trivial cases. We show that this property is invariant under linear equivalence, and derive additional structural properties. In particular, we present a characterization of a subclass of APN polynomials in terms of the $2$-strongly breaking property, which includes the Gold function (see Corollary \ref{cor:2stongly}).

Complementary, the concept of \emph{vanishing $2$-dimensional flats} has been studied in \cite{VanishingFlats2020}, where the authors derived a natural connection with a partial quadruple system that is a substructure of a \emph{Steiner quadruple system.} It was shown that the structures associated to CCZ-equivalent polynomials are isomorphic. More importantly, the authors characterized the number of vanishing flats by means of the differential spectrum. The vanishing property is closely related to a certain code associated to the underlying function (see \cite{CarletBook}, page 412; a thorough analysis has been carried out in \cite{MagdeburgThesis}). In general, the number of vanishing 2-flats is a good measure of how close a function is to an APN function. Namely, for an APN polynomial $F$, the image of the mapping $\Phi_F$ taking a $2$-flat $A$ to $\sum_{x\in A}F(x)$ is included in $\Ftn^*$. Moreover, the $D$-property implies that its image is exactly $\Ftn^*$. The mapping $\Phi_F$ was first used to characterize the AB property in terms of balancedness \cite{CarletBook}. In \cite{DillonsProperty}, the authors obtained some bounds on the parameters $m, n$ of an $(n,m)$-function satisfying the $D$-property by means of such mapping, whereas the authors of \cite{Kaspers2026} used it in their study of Grassmannian partitions related to sum-free functions. In the second part of the paper, we employ this mapping and introduce character sums useful to analyze its behaviour for any polynomial $F$. As proof of concept, we improve Corollary 1 in \cite{BeCaChLa06} by yielding the exact value of the total sum of the sum-of-square indicators of a function. This enables us to easily solve Open Problem 4 in \cite{BeCaChLa06} (see Remark \ref{rem:solution_OP4}). Afterwards, we introduce some natural balancedness properties related to these mappings, which we termed the \emph{balanced and balanced-like properties}. The first notion characterizes \emph{component-wise APNness (CAPNness)} \cite{CarletCAPN18}, where an $(n,n)$-function is component-wise APN if %$v\cdot F$, for $v\in\Ftn^*$,
the fourth power moment of the Walsh spectrum of every nonzero component function is equal to $2^{3n+1}$, $n$ odd. The second notion leads to an extension of this concept to any dimension. By considering further additional balanced partitions of the space, we introduce the \emph{$k$-balanced property} associated to APN functions. Notably, this property is satisfied by the Gold functions, as well as, Dillon's APN permutation. We then establish the EA-invariance of these properties. Several questions remain open, and we hope to stimulate further research on these topics.

The rest of the paper is organized as follows. In Section \ref{sec:preliminaries}, we recollect all definitions and fix the notation that we use along the paper. The notion of $k$-breaking, its structural properties and the number of unbroken flats of a class of functions (including APN functions) are given in Section \ref{def:breaking}. A lower bound on the vectorial nonlinearity is derived in Section \ref{sec:vectorial_nonlinearity}. An in-depth analysis of the types of unbroken flats for mappings whose preimage sets are bounded by $3$ is given in Section \ref{sec:preimages}. In Section \ref{sec:strongly_nonnormal_and_sum_free}, we recall the notions of strongly non-normality and sum-freedom and show that $k$-breaking implies strongly non-normality. Additionally, it is observed that a $k$th-order sum-free APN function is necessarily $k$-breaking for $k>2$. The strongly breaking property is introduced in Section \ref{sec:strongly_breaking}, where we analyze some structural properties and provide a characterization of a subclass of APN functions in terms of $2$-strongly breaking. A character transformation associated to a polynomial is provided in Section \ref{sec:balanced}, where we show that 
this transformation is closely related to the sum-of-square indicators. The solution to Open Problem 4 in \cite{BeCaChLa06} is given in Remark \ref{rem:solution_OP4}. The concepts of balancedness, balanced-likeness and $k$-balancedness are introduced, as well as the characterization of CAPN functions. Finally, we draw some conclusions in Section \ref{sec:conc}.

\section{Preliminaries}\label{sec:preliminaries}
Throughout the paper, we identify the finite field $\F_{2^n}$ with the vector space $\F_2^n$, by fixing a basis of $\F_{2^n}$ over $\F_2$. An $(n,m)$-function is a mapping $F:\Ftn\to \mathbb{F}_{2^m}$, when $m=1$, we refer to $F$ as a Boolean function. The \emph{Walsh-Hadamard transform} of $F:\F_{2^n} \to \F_{2^m}$ at $(u,v) \in \mathbb{F}_{2^m}^*\times \mathbb{F}_{2^n} $ is defined as \begin{equation} W_F(u,v) = \sum\limits_{x \in\Ftn} (-1)^{\operatorname{Tr}_m(uF(x)) + \operatorname{Tr}_n(v x)},\end{equation} where $\operatorname{Tr}_n$ denotes the \textit{absolute trace function} over $\F_{2^n}$ defined by $\operatorname{Tr}_n(x) = x+ x^{2} + \cdots + x^{2^{n-1}}.$ The multi-set of Walsh values $\{* W_F(u,v) : (u,v) \in \mathbb{F}_{2^m}^*\times \mathbb{F}_{2^n} *\}$ is called the \textit{Walsh spectrum} of $F$, whereas the extended Walsh spectrum of $F$ is given by their absolute values. When considering a Boolean function $f:\Ftn\to \F_2$, we will use the shorthand notation $W_f(v) := W_f(1,v)$. For $u\in \mathbb{F}_{2^m}$ and $F:\Ftn\to \mathbb{F}_{2^m}$, the Boolean functions $F_u$, defined by $F_u(x) = \operatorname{Tr}_m (uF(x))$, are called the \emph{components} of $F$. The \emph{nonlinearity} of $F:\Ftn \to \Ftm$
is defined as the minimum nonlinearity of its
component functions, i.e.,
$nl(F) = 2^{n-1} - \frac{1}{2}\max\left\{ |W_F(u,v)| :(u,v) \in \mathbb{F}_{2^m}^*\times \mathbb{F}_{2^n}\right\}.$ For an $(n,m)$-function $F$ and $a \in \Ftn^*$, the \textit{first-order derivative of $F$ at $a$} is the function $F(x+a)+F(x)$. In general, we will be interested in polynomials, i.e., $(n,n)$-mappings, thus every notion will be presented in this setting, even though most of the results can be adapted for general $(n,m)$-functions.

A function $F : \mathbb{F}_{2^n} \to \mathbb{F}_{2^n}$ is \textit{$k$-to-$1$} if every element of $\mathrm{Im}(F)$ has exactly $k$ preimages, where $\mathrm{Im}(F)$ denotes the image set of $F$. The function $F$ is \textit{almost $k$-to-$1$} if exactly one image element has one preimage and all others have $k$ preimages. Given a polynomial $F$ over $\Ftn$ and $(a,b) \in \Ftn^*\times \Ftn$, we define
$$
\de_F(a,b)=|\{ x \in \Ftn \colon F(x+a)+F(x)=b \}|,
$$
where $|\cdot|$ will always denote the cardinality of a set. The multiset of the values $\de_F(a,b)$ is called \emph{differential spectrum}, and the table displaying them is called \emph{difference distribution table (DDT)}. Clearly, $\de_F(a,b)$ must be even (since, if $x_0$ is a solution of the above equation, so is $x_0+a$). The \emph{differential uniformity} of $F$ is defined as
$
\de_F=\max_{a \in \Ftn^*, b \in \Ftn} \de_F(a,b).
$
The polynomial $F$ is called {\em almost perfect nonlinear} (APN), if its differential uniformity is $\de_F=2$, which is the smallest possible. Equivalently, $F$ is APN if and only if the first-order derivative $F(x+a)+F(x)$ is a $2$-to-$1$ mapping for each $a \in \Ftn^*$. A \textit{$k$-dimensional flat} (or $k$-flat) is a $k$-dimensional affine subspace of $\F_{2^n}$ over $\F_2$. The elements in $\Ftn$ are the $0$-dimensional flats and $1$-dimensional flats are all subsets of $\Ftn$ with two points. When dimension is not relevant or when it is obvious from the context, simply \emph{flat} will be used. For $n \ge 2$, define the set of all $2$-dimensional flats in $\Ftn$ as
\begin{equation}\label{eq:2-flat}
	\Tf=\{ \{x_1,x_2,x_3,  x_4\}  \colon \mbox{$\sum_{i=1}^4 x_i=0$ } 
	 \mbox{and $x_i\in\Ftn$ are pairwise distinct} \}.
\end{equation}
Conventionally, each subset in $\Tf$ is called a \emph{block}. The classical \emph{Steiner quadruple system} is a pair $(\Ftn,\Tf)$, so that each $3$-subset of $\Ftn$ is contained in exactly one block of $\Tf$ and every two points of $\Ftn$ is contained in $2^{n-1}-1$ blocks. Thus, the set of $2$-flats yields a $2-(2^n, 4, 2^{n-1}-1)$ design. In this context, it can be shown that a function $F: \Ftn \rightarrow \Ftn$ is APN if and only if for each $\{x_1,x_2,x_3,x_4\} \in \Tf$,
$$
F(x_1)+F(x_2)+F(x_3)+F(x_4) \ne 0.
$$
Namely, the summation of $F$ over each $2$-dimensional flat is non-vanishing. For an arbitrary function $F: \Ftn \rightarrow \Ftn$, define the set of \emph{vanishing flats} as\footnote{The set $\mathcal{Z}_{n,F}$ was denoted by $\mathcal{VB}_{n,F}$ in \cite{VanishingFlats2020}.}
\begin{equation*}
	\mathcal{Z}_{n,F} =\{ \{x_1,x_2, x_3, x_4\} \in \Tf \colon 
	  F(x_1)+F(x_2)+F(x_3)+F(x_4)=0 \}.
\end{equation*}
It can be shown \cite{VanishingFlats2020} that  $|\mathcal{Z}_{n,F}|=\frac{1}{3}\sum_{a \in \Ftn^*, b \in \Ftn} \binom{\de_F(a,b)/2}{2}.$
For a Boolean function $f:\Ftn\to \mathbb{F}_2,$ the quantity $\sum_{a\in\Ftn}W_{D_a f}^2(0)$ is referred to as the \emph{sum-of-square indicator}, and denoted by $\nu(f)$. It is well-known \cite{BeCaChLa06} that $\nu(f)=2^{-n}\sum_{a\in\Ftn} W_f^4(a)$ and that for any function $F:\Ftn\to \Ftn$, $$\sum_{v\in\Ftn^*} \nu(F_v) \geq (2^n-1)2^{2n+1},$$ where the equality is attained if and only if $F$ is APN. A function is called \emph{component-wise APN} if $\sum_{a\in\Ftn} W_{F_v}^4(a)=2^{3n+1}$, for each nonzero $v\in \Ftn$. These functions were introduced and studied in \cite{CarletCAPN18}. Let $n$ and $s$ be two integers with the same parity. A Boolean function $f:\F_{2^n}\to \F_{2}$ is called \textit{plateaued} with amplitude $2^s$ if its Walsh coefficients take on all the values in $\{0,\pm 2^{\frac{n+s}{2}}\}$. When $s=0$, the function is said to be \emph{bent}. In this case, $n$ must be even and the Walsh coefficient $0$ is never attained. Two functions $F_1, F_2: \Ftn\to \mathbb{F}_{2^m}$ are called \textit{affine equivalent} (resp. \emph{linear equivalent}) if there exist two affine (resp. linear) automorphisms  $L, L'$ on $\F_{2^n}$ and $\mathbb{F}_{2^m}$, respectively, such that $F_2= L'\circ F_1 \circ L$. If there also exists an affine function $L'':\F_{2^n}\to \F_{2^m}$ such that $F_2=L'\circ F_1\circ L+L''$ then they are called \textit{extended affine equivalent} (EA, for short). Moreover, if there exists an affine automorphism on $\F_{2^n}\times \mathbb{F}_{2^m}$ that maps the set $\{(x,y) \in \F_{2^n}\times \mathbb{F}_{2^m} : y=F_1(x) \}$ onto $\{(x,y) \in \F_{2^n}\times \mathbb{F}_{2^m} : y=F_2(x) \}$, then the functions are called \textit{CCZ-equivalent} (named after Carlet, Charpin,  and Zinoviev). Affine equivalence is strictly stronger than EA-equivalence, which in turn is strictly stronger than CCZ-equivalence \cite{CarletVectBound}. 
Throughout the paper, the preimage of an element $a\in \Ftn$ under $F:\Ftn\to \Ftn$ (called a \emph{fiber}) will be denoted by $F^{-1}[a]$, i.e., $$F^{-1}[a]= \{ x\in \Ftn :  F(x)=a\}.$$ For $S\subset \Ftn$, the symbol $F|_S$ is reserved for the restriction of $F:\Ftn\to\Ftn$ to $S$, and the image of $S$ under $F$ will be denoted as $F(S)$, i.e., $F(S) := \mathrm{Im}(F|_S)$. 

\section{Different notions related to APN mappings}

In this section, we study different notions related to the APN property. First, an analysis of the so-called breaking property of polynomial mappings is carried out, where we focus on the $2$-breaking property of APN functions and provide a characterization thereof.  We then present a lower bound on the vectorial nonlinearity of vectorial mappings together with a refinement of the breaking property for functions whose preimage sizes are bounded by three. Afterwards, two important notions, strongly non-normality and sum-freedom, recently introduced by Carlet in \cite{CarletTwoNotions2024}, are described and their links with the breaking property are established. To better understand these intertwined properties, we introduce a relevant concept that will be referred to as \emph{strong breaking}. Some general results are derived and a characterization of a subclass of APN functions in terms of $2$-strongly breaking is obtained.

\subsection{$k$-breaking functions}\label{sec:breaking}

In the context of cryptographically significant mappings, the following natural definition was employed in \cite{KolomeecDCC2024}, where the authors studied properties of the multiplicative inverse permutation.

\begin{defi}\label{def:breaking} We say that a mapping $F \colon \F_{2^n} \to \F_{2^n}$ {\em breaks} a flat $A$ of $\F_{2^n}$ if $F(A)$ is not a flat of $\F_{2^n}$. If $F$ breaks all $k$-dimensional flats, $F$ is said to be \emph{$k$-breaking}\footnote{Generalizations of the negation of the $k$-breaking property were studied in \cite{BraWoPre05,Enrico}.}.
\end{defi}

It is evident that no function can break $0$-dimensional flats or $1$-dimensional flats, hence we will assume that the dimension of $A$ is at least $2$ whenever we discuss the breaking property. Oftentimes, we will refer to the breaking property as a property of the flat by saying that $A$ is (un)broken with respect to $F$. For $F:\Ftn \to \Ftn$, define the set of unbroken $k$-dimensional flats with respect to $F$ as $$\UB_{n,k,F} = \{ A \mbox{ is a }k\mbox{-flat} \ :\ F(A) \text{ is a flat} \}.$$ We will typically omit the parameter $k$ whenever $k=2$, unless otherwise required by the context.

\begin{prop} The property of being $k$-breaking is invariant under affine equivalence, i.e., for any affine automorphisms $L$ and $L'$ of $\mathbb{F}_{2^n}$, $F$ is $k$-breaking iff $L'\circ F\circ L$ is $k$-breaking. Moreover, $$| \mathcal{UB}_{n,k,F} | = |\mathcal{UB}_{n,k,L'\circ F\circ L} |$$ via the correspondence $A \in \mathcal{UB}_{n,k,F} \mapsto L^{-1}(A)\in  \mathcal{UB}_{n,k,L'\circ F \circ L} $.
\end{prop}

\begin{proof} For each $k$-dimensional flat $A=a+V$, where $V$ is a $k$-dimensional subspace and $a\in\Ftn$ and for any affine automorphism $L(x)=  x\cdot M+c$ of $\Ftn$, where $M$ is a non-singular matrix and $c\in\Ftn$, the image set $L(A)$ is again a $k$-dimensional flat since $L(A)= \{ x\cdot M + a\cdot M +c : x\in V\}$. Suppose $F$ is not $k$-breaking. There is a $k$-dimensional flat $A$ such that $F(A)$ is a flat. Since $L^{-1}$ is an affine automorphism, $L^{-1}(A)$ is a $k$-dimensional flat, say, $B$. Compute $L'\circ F\circ L(B) = (L'\circ F )(A)$. As $F(A)$ is a flat and $L'$ is an affine automorphism, we have that $L'(F(A))$ is a flat implying that $L'\circ F\circ L$ is not $k$-breaking. We have shown that, for any function $F:\Ftn\to \Ftn$ and any  affine automorphisms $L,L'$ of $\Ftn$, if $L'\circ F\circ L$ is $k$-breaking then $F$ is $k$-breaking. Considering the function $L'\circ F \circ L$ and the affine automorphisms $L^{-1},(L')^{-1}$, the result follows.    
\end{proof}

\begin{rem}\label{rem:breakingCZZ} In general, the $k$-breaking property is not an extended affine invariant. For instance, the Gold function $F(x)=x^3$ over $\mathbb{F}_{2^5}$ is an APN permutation, so it is $2$-breaking. However, there exist some unbroken $2$-dimensional flats after adding a linear permutation $L(x)$. More precisely, the function $F:\mathbb{F}_{2^5}\to\mathbb{F}_{2^5}$ given by $x\mapsto x^3+L(x)=x^3+x^2$ has the preimage distribution $\{*1^{15}, 2^1, 3^5 *\}$, so there are $5$ unbroken $2$-dimensional flats.  
\end{rem}

\begin{rem} The breaking property is neither upward nor downward monotone in general. For instance, the multiplicative inverse $F(x)=x^{2^n-2}$ over $\Ftn$, $n$ odd, is $2$-breaking as it is an APN permutation. However, for each divisor $d$ of $n$, the subfield $\mathbb{F}_{2^d}$ is invariant, thus $F$ is not $d$-breaking. On the other hand, the APN Gold function $F(x)=x^3$ over $\mathbb{F}_{2^6}$ breaks all $4$-dimensional flats, however, there are unbroken $2$-dimensional and $3$-dimensional flats (see Tables \ref{table: APNs unbroken flats} and \ref{table:unbroken-vanishing-3}). 
\end{rem}

In what follows, we will focus on the breaking property of $2$-dimensional flats with an emphasis on APN functions. Given $a\in\F_{2^n}^*$, define $P_a:=\{\{x,x+a\}\subset\F_{2^n} : D_aF(x)=0\}$. It is easy to see that $P_a\cap P_{a'}=\emptyset$ whenever $a\neq a'$. Indeed, if $\{x,x+a\}=\{x',x'+a'\}$, then either $x=x'$, in which case $a=a'$, or $x=x'+a'$, in which case $x+a=x'$, and therefore $a=(x+a)+x=x'+(x'+a')=a'$.

\begin{prop}\label{prop:collapsing_flat}
    Let $F:\Ftn \to \Ftn$, and let $A\in \Tf$. Suppose $F|_A$ is not 2-to-1. Then $F(A)=\{c\}$, for some $c\in \Ftn$ if and only if $A=B\cup C$, where $B,C\in P_a$ are distinct, for some $a\in \Ftn^*$.
\end{prop}
\begin{proof}
    Let $A=\{x,x+a,x+b,x+a+b\}$, and suppose $F(A) =\{c\}$, for some $c\in \Ftn$. Then,
    $$F(x) + F(x+a) = F(x+b) + F(x+b+a) = 0,$$
    hence $\{x,x+a\}, \{x+b,x+b+a\}\in P_a$.  Now, suppose $A = \{x_1,x_1+a\}\cup \{x_2,x_2+a\}$, where $\{x_1,x_1+a\},\{x_2,x_2+a\} \in P_a$ are distinct. Then,
      $$F(x_1) =F(x_1+a) \text{ and }   F(x_1+b) + F(x_1+b+a) = 0,$$
      and since $F|_A$ is not 2-to-1, it follows $F(x_1) =F(x_1+a) = F(x_2) =F(x_2+a).$
\end{proof}

Given a function $F:\mathbb{F}_{2^n}\to \mathbb{F}_{2^n}$, we say that a set $A\subseteq \mathbb{F}_{2^n}$ is \emph{transversal} to $F$ (or, simply, \emph{$F$-transversal}) if for every $u\in \mathbb{F}_{2^n}$, $|A\cap F^{-1}[u]| \leq 1$. Equivalently, $A$ is $F$-transversal if and only if $F|_A$ is injective. Denote the set of fibers with at least three elements by $\mathcal{P}$, i.e., $$\mathcal{P} := \{F^{-1}[x]\ : \ x\in\Ftn, |F^{-1}[x]|\geq 3 \}.$$
For each $P\in \mathcal{P},$ define the sets
$$\Gamma_P= \{ A\in \Tf : A\subseteq P \} \mbox{ and }\Lambda_P = \{ A\in \Tf : |A\cap P|=3 \}.$$
\begin{theo}\label{th:unbroken_general}
   Let $F:\Ftn\to\Ftn$ be  such that, for every $A\in\Tf$, it holds:
\begin{itemize}
    \item[i)] the restriction $F|_A$ is never $2$-to-$1$; 
    \item[ii)]if $A$ is $F$-transversal, then $\sum_{x\in A}F(x)\neq 0.$
\end{itemize}Then, the number of unbroken $2$-dimensional flats w.r.t. $F$ is 
$$|\UB_{n,F}|= \sum_{P\in \mathcal{P}} (|\Gamma_P|+|\Lambda_P|) = \frac{1}{3}\sum_{a\in\Ftn^*}\binom{|P_a|}{2} + \sum_{P\in \mathcal{P}} |\Lambda_P|.$$
In particular, if $\Gamma_P=\emptyset$ for all $P\in \mathcal{P}$, then  $|\UB_{n,F}|=\sum_{P\in \mathcal{P}}\binom{ |P|}{3}.$
    \end{theo}
 \begin{proof}
Let $A=\{x_1,x_2,x_3,x_4\}$ be a $2$-dimensional flat. If $A$ is $F$-transversal, i.e., $F(x_1)=a$, $F(x_2)=b$, $F(x_3)=c$, and $F(x_4)=d$, where $a,b,c,d\in\Ftn$ are pairwise distinct, then $\sum_{i=1}^4 F(x_i)\neq 0$ by hypothesis. Hence, $A$ is broken. If $F(x_1)=a$, $F(x_2)=b$, and $F(x_3)=F(x_4)=c$, where $a,b,c\in\Ftn$ are pairwise distinct, then the flat $A$ is broken since $|F(A)|=3$. Therefore, the only possible unbroken $2$-flats are such that $F$ is either constant or almost 3-to-one on them. This implies that any unbroken flat $A$ must either be contained in some $P$ in $\mathcal{P}$ or intersect a unique $P$ in $\mathcal{P}$ in three points. In other words, 
$ \UB_{n,F}= \bigcup_{P\in \mathcal{P}} (\Gamma_P\cup \Lambda_P).$ Since all of the terms in the unions are disjoint, we have $$|\UB_{n,F}|= \sum_{P\in \mathcal{P}} (|\Gamma_P|+|\Lambda_P|).$$ Moreover, if $F(x_1)=F(x_2)=F(x_3)=F(x_4)=c$ for some $c\in\Ftn$. Then $F(A)=\{c\}$ is a $0$-flat. By Proposition~\ref{prop:collapsing_flat}, it follows that
$A=\{x_1,x_1+a\}\cup\{x_2,x_2+a\},$ where $\{x_1,x_1+a\}, \{x_2,x_2+a\}\in P_a$, $x_1\neq x_2$. Thus, $\sum_{P\in \mathcal{P}} |\Gamma_P| = \frac{1}{3}\sum_{a\in\Ftn^*}\binom{|P_a|}{2}$ since every direction is counted three times.
The last part follows from the fact that if $\Gamma_p = \emptyset$ for each $P\in\mathcal{P}$ (equivalently, $|\mathcal{P}_a| \leq 1$ for each $a\in \Ftn^*$), then every three points in $P$ determine an unbroken flat, so that $|\UB_{n,F}|=\sum_{P\in \mathcal{P}}\binom{ |P|}{3}.$
\end{proof}

Observe that if $F$ satisfies the conditions of Theorem \ref{th:unbroken_general}, it is almost $k$-to-one and $F$ does not ``collapse'' any $2$-flat to a single point, then, the number of unbroken 2-flats is $$\UB_{n,F} = \frac{\binom{k}{3}(2^n-1)}{k}=\frac{(2^n-1)(k-1)(k-2)}{6}$$ since every fiber (except for one) satisfies $|F^{-1}[a]| = k$ , and there are $\frac{2^n-1}{k}$ such fibers. 

\begin{cor}\label{th:breakingAPN}
    Let $F\colon\F_{2^n}\to \F_{2^n}$ be an APN function. If $\mathcal{P}\not=\emptyset$, then the number of unbroken $2$-dimensional flats is $|\UB_{n,F}|=\sum_{P\in \mathcal{P}}\binom{ |P|}{3}.$ Otherwise,  $|\UB_{n,F}|=0$. Conversely, for any function $F\colon\F_{2^n}\to \F_{2^n}$, if either $|\UB_{n,F}|=0,$ when $\mathcal{P}=\emptyset,$ or, otherwise, $|\UB_{n,F}|=\sum_{P\in \mathcal{P}}\binom{ |P|}{3},$ and $|F(A)| > 1$ for each $A\in\Tf$, then $F$ is APN.
\end{cor}
\begin{proof} 
    Note that for a 2-flat $A=\{x_1,x_2,x_3,x_4\}$, it must hold $\sum_{i=1}^4F(x_i)\neq 0$, since $F$ is APN. So the restriction $F\vert_A$ of $F$ to $A$ is neither constant nor 2-to-1. The result follows at once from Theorem \ref{th:unbroken_general}. Conversely, suppose that the number of unbroken $2$-flats is as given in the statement. Suppose first that $\mathcal{P}=\emptyset$ and $|\UB_{n,F}|=0.$ Let $A=\{x_1,x_2,x_3,x_4\}$ be any $2$-flat. The only possible ``configurations'' for $2$-flats w.r.t $F$ are: $F(x_1)=F(x_2)=a$ and $F(x_3)=b$ and $F(x_4)=c$, for some distinct points $a,b,c\in \Ftn$ or $A$ is $F$-transversal. In the first case, $\sum_{x\in A} F(x) \not=0$ since it is the sum of two different points and in the second case 
     $\sum_{x\in A} F(x) \not=0$ since the image cannot be a flat.
    Suppose that $\mathcal{P}\not=\emptyset$, $|F(A)| > 1$, for each $A\in\Tf$, and  $|\UB_{n,F}|=\sum_{P\in \mathcal{P}}\binom{ |P|}{3}$. From here, we see that every element in $\mathcal{P}$ determines exactly $\sum_{P\in \mathcal{P}}\binom{ |P|}{3}$ unbroken flats. Thus, all other flats must be broken. Let $A=\{x_1,x_2,x_3,x_4\}$ be any $2$-flat. If $F|_A$ is $3$-to-$1$ then $\sum_{x\in A}F(x)$ is the sum of two distinct points, so it is non-null. Similarly, if $F(x_1)=F(x_2)=a$ and $F(x_3)=b$ and $F(x_4)=c$, $\sum_{x\in A} F(x) = b+c \not=0$. The function $F$ cannot be $2$-to-$1$ on $A$ since $A$ would be an unbroken flat. Finally, if $A$ is $F$-transversal, then $\sum_{x\in A} F(x) \not = 0$ as $F(A)$ cannot be a flat.
    \end{proof}

\begin{rem}\label{th:breaking3-to-1} 
    Let $n$ be odd and let $F\colon\F_{2^n}\to \F_{2^n}$ be an almost 3-to-1 function.\footnote{Recall that there are no APN functions which are almost $k$-to-1, for $k\geq 4$, and, there are no APN functions which are $k$-to-1 for $k>2$. Indeed, it is well-known \cite{GoharDCC2023} that the image set of $F$ cannot be small, i.e., $|\mathrm{Im}(F)|\geq \frac{2^n+1}{3},$ if $n$ is odd, and $|\mathrm{Im}(F)|\geq \frac{2^n+2}{3},$ if $n$ is even.
}
From Corollary \ref{th:breakingAPN}, it follows immediately that $F$ is APN if and only if the number of unbroken $2$-dimensional flats is $|\UB_{n,F}|=\frac{(2^n-1)}{3}.$ Similarly, when $F$ is a function whose all preimages have size at most $2$, then $F$ is APN if and only if $|\UB_{n,F}|=0.$
\end{rem}

\begin{ex}
   Let $F: \mathbb{F}_{2^6} \to \mathbb{F}_{2^6}$ be the Gold function defined by $F(x) = x^3$. Since $F$ is almost $3$-to-$1$, Corollary~\ref{th:breakingAPN} implies that $|\UB_{6,F}| = 21$; in other words, exactly $21$ two-dimensional flats are unbroken with respect to $F$. The set $\UB_{6,F}$ includes only subspaces and is given by
    \begin{align*}
        \{
 &\{38, 0, 24, 62\},\{21, 46, 0, 59\},\{29, 0, 16, 13\}, \{40, 0, 49, 25\},\{51, 0, 54, 5\},\{27, 47, 0, 52\},
 \{19, 0, 31, 12\},\\
&
 \{55, 0, 61, 10\},
 \{35, 8, 0, 43\}, \{23, 0, 39, 48\},
 \{22, 41, 63, 0\},
 \{50, 57, 0, 11\}, \{14, 15, 0, 1\},
 \{30, 0, 2, 28\},
 \\
 &
 \{3, 0, 18, 17\}, \{37, 0, 9, 44\},\{56, 60, 0, 4\},\{0, 33, 53, 20\} \{0, 32, 58, 26\},\{7, 0, 42, 45\},\{36, 0, 6, 34\}
\},
    \end{align*}
   where the field element $\alpha_0+\alpha_1\omega+\alpha_2\omega^2+\alpha_3\omega^3+ \alpha_4\omega^4+ \alpha_5\omega^5$ corresponds to the integer $\sum_{i=0}^{5}\alpha_{i} 2^i$ for a root $\omega$ of the binary primitive polynomial $x^6+x^4+x^3+x+1$.
\end{ex}

Note that Corollary \ref{th:breakingAPN} says that the only unbroken $2$-dimensional flats of an APN function map to $1$-dimensional flats. In general, for an APN function, a stronger property is true for any $k>2$. Albeit the following result is well-known, to the best of our knowledge, there is not an explicit proof of it. Its simple proof is derived for the sake of completeness.

\begin{prop} \label{prop:Sadmir_prop} Let $F:\Ftn\to \Ftn$ be an APN function. If $A\in \mathcal{UB}_{n,k,F}$ and $F(A)$ is an $l$-flat, then $l=k$ or $k=2$ and $l=1$. In particular, every APN function is $4$-breaking.
\end{prop}

\begin{proof} Suppose that $A=a+V$, where $V$ is a $k$-dimensional subspace and $a\not\in V$. First assume that $0\in F(A)$. Define $G:V\to F(A)$ by $G(x) = F(a+x)$. Given distinct $x,y,z\in V$, we have that $\{a+x,a+y,a+z,a+x+y+z\} \in \Tf$, so it must be that $G(x)+G(y)+G(z)+G(x+y+z)=F(a+x)+F(a+y)+F(a+z)+F(a+x+y+z)\not=0.$ This means that $G$ is an APN function mapping from a $k$-dimensional space onto an $l$-dimensional space. 
Nyberg's results \cite{Nyberg94} imply that an APN $(k,l)$-mapping exists only for $k=l$ or $k=2, l=1$. Now, suppose that $0\not\in F(A)$. Define $G:V\to W$ by $G(x) = F(a+x)+F(a)$, where $W=F(A)+F(a)$ is an $l$-dimensional subspace since $F(a)\in F(A)$. Similarly as before, given distinct $x,y,z\in V$, we have that $\{a+x,a+y,a+z,a+x+y+z\} \in \Tf$, so it must be that $G(x)+G(y)+G(z)+G(x+y+z)=F(a+x)+F(a+y)+F(a+z)+F(a+x+y+z)\not=0.$ Again, $G$ is an APN function mapping from a $k$-dimensional space onto an $l$-dimensional space, therefore $k=l$ or $k=2, l=1$. The last statement follows from the non-existence of APN permutations over $\F_{2^4}.$
\end{proof}

Note that Proposition \ref{prop:Sadmir_prop} implies that a quadratic APN function is $k$-breaking for all even $k>2$ and any cubic APN function is $6$-breaking since the restriction of a function to a flat cannot increase their degree. Moreover, Proposition 3.5 in \cite{Enrico} implies that the nonlinearity $nl(F)$ is bounded above by $$ 2^{n-1} -  \max_{\substack{2\leq k < n, \\ F {\footnotesize \mbox{ is not }}k{\footnotesize \mbox{-breaking}}}} 2^{k-1},$$ thus if $\max_{a\not=0,b\in \Ftn} |W_F(a,b)| < 2^l$, for some $l$ with $2 \leq l <n$, then $F$ is $l$-breaking. Therefore, it is natural to expect an APN function to break all large dimensional flats. For instance, according to \cite{BraWoPre05,KolomeecDCC2024}, the multiplicative inverse breaks all flats except for cosets of subfields. Thus, for $n$ odd, it breaks all $k$-flats for $k\geq \lfloor \frac{n}{3} \rfloor +1$. Moreover, an AB function satisfies $\max_{a\not=0,b\in \Ftn} |W_F(a,b)| = 2^{\frac{n+1}{2}}$, thus it must be $l$-breaking for each $l > \frac{n+1}{2}$. Every APN permutation is $(n-1)$-breaking. However, for non-bijective APN functions, the exact bound is not so clear and the  $(n-1)$-breaking property does not necessarily hold, as illustrated by the following example.

\begin{ex} The function $T_6:\F_{2^6}\to\F_{2^6}$
\cite[Section~5.3]{trims22} given by
\[
\begin{pmatrix}x\\y\end{pmatrix}
\longmapsto
\begin{pmatrix}x^3\\0\end{pmatrix}
+
\begin{pmatrix}x^{16}+x\\
\operatorname{Tr}_{n}(x)
\end{pmatrix}y,
\qquad x\in\F_{2^5},\ y\in\F_2,
\]
 is a quadratic APN function with maximum possible linearity, obtained as an extension of the Gold
function $x\mapsto x^3$ on $\F_{2^5}$. It is easy to see that $T_6$ is not $5$-breaking.  Since $x\mapsto x^3$ is a permutation of $\F_{2^5}$, for the 5-flat $A=\F_{2^5}\times \{0\}$ we have $T_6(A)=\{(x^3,0):x\in\F_{2^5}\}=\F_{2^5}\times\{0\}=A.$
\end{ex}
 It is worth noting that in \cite{BeierleCarlet2023}, the authors proved that one cannot have a similar extension of Gold for dimensions $n>5$.  This naturally leads to the following question.
\begin{op}
    For $n\geq 7$, is every APN function $F$ over $\mathbb{F}_{2^n}$ $k$-breaking for $ \lfloor \frac{n}{2} \rfloor+1 \leq k \leq n-1$?
\end{op}

\subsubsection{Lower bound on vectorial nonlinearity via breaking property}\label{sec:vectorial_nonlinearity}

In this section, we provide a lower bound on vectorial nonlinearity, introduced in \cite{LiuChenMesn2017}, by employing the cardinality of broken/unbroken 2-dimensional flats. More precisely, the vectorial nonlinearity of $F:\F_{2^n} \to \F_{2^m}$ is
\begin{equation}
\mathcal{NL}_V F =\min_{L \in \mathcal{A}_{n,m}} d_H(F,L), 
\end{equation}
where $\mathcal{A}_{n,m}$ denotes the set of affine functions from $\F_{2^n}$ to $\F_{2^m}$, and $d_H(F,L) = | \{ x \in \F_{2^n} : F(x) \neq L(x) \}|$. 
The main idea is to simply use the estimate on the number  of broken/unbroken 2-dimensional flats, which then gives the lower bound.
\begin{theo}\label{th:vectnonlinearity}
Let $F:\Ftn \rightarrow \Ftn$ be any polynomial with $|\UB_{n,F}| = w$. Then, $\mathcal{NL}_V F \geq 2^{n-2}-w.$
In particular, if $F$ is $2$-breaking, then $\mathcal{NL}_V F \geq 2^{n-2}$. 
\end{theo}
\begin{proof}
	Represent $\Ftn$ as a union of 2-dimensional flats, thus $\Ftn=\cup_{i=1}^{2^{n-2}}A_i$, where $A_i$ are disjoint 2-dimensional  flats. Then, for any $L \in \mathcal{A}_{n,n}$, its restriction $L|_{A_i}$ to $A_i$ is also affine. Therefore 
	$$\mathcal{NL}_V F  =\min_{L  \in \mathcal{A}_{n,n}} \left(\sum_{i=1}^{2^{n-2}} d_H(F|_{A_i},L|_{A_i})\right).$$
    If $A_i$ is broken, then $F|_{A_i}$ is not affine. Hence, $d_H(F|_{A_i},L|_{A_i})\geq 1$. This readily implies that  $\min_{L  \in \mathcal{A}_{n,n}} \left(\sum_{i=1}^{2^{n-2}} d_H(F|_{A_i},L|_{A_i})\right) \geq 2^{n-2}-w.$    In particular, when $|\UB_{n,F}| = 0$ then $\mathcal{NL}_V F \geq 2^{n-2}$.
\end{proof}

\begin{ex} Let $F:\Ftn\to \Ftn$ be any APN permutation. By Corollary \ref{th:breakingAPN}, $F$ breaks all 2-dimensional flats. This implies that $\mathcal{NL}_V F \geq 2^{n-2}$ by Theorem \ref{th:vectnonlinearity}.
\end{ex} 

\begin{rem} The lower bound obtained in Theorem \ref{th:vectnonlinearity} is rather loose for most cases, though, sometimes, it can lead to better estimates than previously known bounds. The lower bound in Theorem 4.5 of \cite{LiuChenMesn2017} states that $\mathcal{NL}_V F \geq 2^n - \max\{ d(F), 2^{n-1}\}$, where $d(F)$ is the polynomial degree of $F$. When $d(F)$ is strictly larger than $2^n-2^{n-2} +w$, our bound yields an improvement. For instance, $F(x)=x^{2^n-2}$ satisfies $2^n - \max\{ d(F), 2^{n-1}\} = 2<2^{n-2}$.
\end{rem}

\begin{rem} The problem of computing vectorial nonlinearity is considered difficult (see \cite{LiuChenMesn2017, Nagy2025}.  Carlet, in  a private communication with Nagy \cite{CarletVectBound}, suggested the lower bound on vectorial nonlinearity $\mathcal{NL}_V F \geq 2^{n} - \sqrt{2^n + \delta_F(2^{n}-1)}$. This bound was recently improved by Nagy using the fact that the fibers of any vectorial map $F:\Ftn\to \mathbb{F}_{2^m}$ are \emph{$\delta_F$-thin sets}\footnote{ A $t$-thin set $T$ is a subset of $\Ftn$ such that $|T\cap (T+a)|\leq t$ for each $a\in\Ftn$.}, yielding $\mathcal{NL}_V F \geq 2^{n} - \sqrt{\delta_F}2^{n/2} -1/2.$ 
\end{rem}

The bound in Theorem \ref{th:vectnonlinearity} can possibly be improved, for APN functions, by further employing the structure of 2-flat decomposition. Especially, for quadratic APN functions. 

\begin{op} Provide a tight lower bound of $\mathcal{NL}_V F$, where $F$ is a quadratic APN function.
\end{op}

\subsubsection{Unbroken flats of functions with bounded preimages}\label{sec:preimages}

For any function $F\colon \F_{2^n}\to \F_{2^n}$ and a $k$-dimensional flat $A$ of $\F_{2^n}$, we define the numbers $\lambda_{A,i}$, $0\leq i\leq 2^k$ to be the number of preimages of $F$ that intersect $A$ in $i$ elements, i.e., \begin{equation}\label{eq:lambdas_def}
\lambda_{A,i} := |\{ b\in \F_{2^n} :  |F^{-1}[b]\cap A| = i \}|.
\end{equation}
Suppose the image of $A$ is again a flat, say, $l$-dimensional. Then,  
\begin{equation}\label{eq:lambdas}
	\sum_{i=1}^{2^k} i\lambda_{A,i} = 2^k\ \mbox{ and }\ \sum_{i=1}^{2^k} \lambda_{A,i} = 2^l.
\end{equation} 
These two equations constrain the possible values that the $\lambda$'s can take. For instance, the first equation already gives that there must be an even number of odd indices $i$ such that $\lambda_{A,i} \not=0$. Indeed, since $ \sum_{i=1}^{2^k} i\lambda_{A,i} \equiv 0 \pmod{2}$, then $\lambda_{A,1}+\lambda_{A,3}+\cdots+\lambda_{A,2^k-1} \equiv 0$, so it must be that there is an even number of odd-indexed $\lambda_{A,i}$ with odd value.

\begin{lemma}\label{lemma_lambdas}Let $F\colon \F_{2^n}\to \F_{2^n}$ be a function such that all its preimage sets have size at most $b$. If $A$ is a $k$-dimensional flat with $k>1$, whose image is an $l$-dimensional flat, $0\leq l\leq k$, then $b\geq 2^{k-l}.$ In particular, if $b \leq 2^e-1$, $e\in \mathbb{N}$, then $l\geq k-e+1$.
\end{lemma} 
\begin{proof}
From  \eqref{eq:lambdas}, we have $\sum_{i=1}^{b} i\lambda_{A,i} = 2^k$ and $\sum_{i=1}^{b} \lambda_{A,i} = 2^l$. Expressing both equations in terms of $\lambda_{A,b}$ and equalizing,
$$\lambda_{A,b} = \frac{1}{b}(2^k-\sum_{i=1}^{b-1} i\lambda_{A,i}) = 2^l-\sum_{i=1}^{b-1} \lambda_{A,i}.$$ This readily implies $\sum_{i=1}^{b-1} (b-i)\lambda_{A,i} = b\cdot 2^{l} - 2^k$. Since $\lambda_{A,i}\geq0$ for each $i\in \{1,\ldots,b-1\}$, it follows that $2^{k-l}\leq b.$ Finally, if $b\leq 2^e-1$ for some positive integer $e$, then, by the above, $2^{k-l} \leq b \leq 2^{e}-1<2^e$. Thus, $2^{k-l+1}\leq 2^e$, which implies $k-e+1\leq l$. 
\end{proof}

Lemma \ref{lemma_lambdas} implies that for $b=3$, the only possible unbroken $k$-dimensional flats map to a $k$-dimensional flat or to a $(k-1)$-dimensional flat. In this case, we can bound the possible values for each $\lambda_{A,i}$ as follows.

\begin{prop}\label{prop:boundslambda} Let $F\colon \F_{2^n}\to \F_{2^n}$ be a function such that all its preimage sets have size at most 3. Let $A$ be a $k$-dimensional flat with $k>1$, whose image is a $(k-1)$-dimensional flat. Then,
   $$\lambda_{A,3}=\lambda_{A,1} \mbox{ and } \lambda_{A,2} = 2^{k-1}-2\lambda_{A,1}.$$ Moreover, $0\leq\lambda_{A_1} = \lambda_{A,3} \leq 2^{k-2}, 0\leq\lambda_{A,2}\leq 2^{k-1}$.
\end{prop}

\begin{proof} Consider the equalities in \ref{eq:lambdas}) given by
\begin{equation}\label{eq:0lambdas30}
    \lambda_{A,1} + 2\lambda_{A,2} + 3\lambda_{A_3} = 2^k
\mbox{ and }
    \lambda_{A,1} + \lambda_{A,2} + \lambda_{A,3} = 2^{k-1}.
\end{equation}

Let us express the system of equations in terms of  $\lambda_{A,1}$. Thus,  
\begin{equation}\label{eq:system}
\lambda_{A,2} = 2^{k-1}- 2 \lambda_{A,1} \mbox{ and }\lambda_{A,3} = \lambda_{A,1}.
\end{equation}

It always holds $\lambda_{A,2}\leq 2^{k-1}$. Now, since $\lambda_{A,2}\geq0$, the first equation in \eqref{eq:system} implies $\lambda_{A,1} \leq 2^{k-2}$ and we always have $\lambda_{A,2}\leq 2^{k-1}$.
\end{proof}

\begin{theo}\label{theo:lambdas}Let $F\colon \F_{2^n}\to \F_{2^n}$ be a function such that all its preimage sets have size at most 3. Let $A$ be a $k$-dimensional flat whose image is an $l$-dimensional flat. 
\begin{enumerate}[i)]
    \item If $k=2$, then either $A$ is $F$-transversal, $\lambda_{A,1} = \lambda_{A_3}=1$ or $\lambda_{A_2}=2, \lambda_{A_1}=\lambda_{A_3}=0$.
    \item  For $k>2$, the numbers $\lambda_{A,1}, \lambda_{A,2}$ and $\lambda_{A,3}$ can have exactly one of the following patterns $(\lambda_{A,1},\lambda_{A,2},\lambda_{A,3})$ modulo 4: $(0,0,0)$,$(2,0,2)$, $(1,2,1)$, or $(-1,2,-1)$. More precisely, the only a priori unbroken flats can be be either $F$-transversal or have the patterns $(\lambda_{A,1},\lambda_{A,2},\lambda_{A,3})$ given by
$$ (0,2^{k-1}, 0), (1, 2^{k-1}-2,1), \ldots, (i,2^{k-1}-2i,i),\ldots, (2^{k-2},0,2^{k-2}).$$
\end{enumerate}
\end{theo}

\begin{proof} If $k=2$ and $A$ is not $F$-transversal, then by Proposition \ref{prop:boundslambda}, $\lambda_{A,1} = \lambda_{A,3} \in \{0,1\}$. If $\lambda_{A,1}=\lambda_{A,3} =0$ then $\lambda_{A,2} = 2$, whereas if $\lambda_{A_1}=\lambda_{A_3}=1,$ then $\lambda_{A,2} =0$. Now, suppose $k>2$.
If $A$ is $F$-transversal, then $\lambda_{A,1} = 2^k,  \lambda_{A,2}=\lambda_{A,3}=0$ yielding the pattern $(0,0,0)$ modulo 4. If $A$ is not $F$-transversal, Proposition \ref{prop:boundslambda} gives that $\lambda_{A,1}=\lambda_{A,3}$ and $\lambda_{A,2}= 2^{k-1}-2\lambda_{A,1}$. Therefore, $\lambda_{A,1} \equiv \lambda_{A,3} \pmod{4}\ \mbox{ and }\ \lambda_{A,2} \equiv 2\lambda_{A,1} \pmod{4},$ which fully determine the possible patterns modulo 4. Taking $\lambda_{A,1}$ as a free parameter and using Proposition \ref{prop:boundslambda} once more, we obtain the desired patterns for the preimage distributions.
\end{proof}

As an interesting consequence of Theorem \ref{theo:lambdas}, we can completely determine the possible types of the preimages of an unbroken flat for the case $k=3$.

\begin{cor} For a function $F\colon \F_{2^n}\to \F_{2^n}$ whose all preimage sets have size at most 3, the only possible unbroken $3$-dimensional flats satisfy $$(\lambda_{A,1}, \lambda_{A,2}, \lambda_{A,3})\in \{(8,0,0), (0,4,0), (2,0,2), (1,2,1) \}.$$
\end{cor}

\begin{rem} By Lemma \ref{lemma_lambdas}, we see that if the cardinalities of preimages of a function $F$ are bounded above by $7$, then the only possible unbroken $k$-flats map to $l$-flats, for $l\in\{k-2,k-1,k\}$. Thus, one can derive similar results to Proposition \ref{prop:boundslambda} and Theorem \ref{theo:lambdas}. However, the systems of equations become more cumbersome compared to the case $b = 3$.
\end{rem}

\subsection{$k$-strongly nonnormal and $k$-th-order sum-free functions} \label{sec:strongly_nonnormal_and_sum_free}

In a recent article,  Carlet introduced \cite{CarletTwoNotions2024} two notions related to properties of $(n,m)$-\textit{functions} $F:\F_{2^n} \to \F_{2^m}$. Precisely, we have the following definitions. 

\begin{defi}
	Let $2 \leq  k \leq n$ and $m$ be positive integers. An $(n, m)$-function $F$
	is called \emph{$k$-strongly non-normal} (resp. \emph{$k$th-order sum-free}) if, for every  $k$-flat $A$ of $\F_2^n$ , the restriction of $F$ to $A$ is not an
	affine function (resp. the sum $\sum_{x\in A} F(x)$ is nonzero). 
\end{defi}
It is easy to see that $k$th-order sum-freedom implies $k$-strong non-normality since the sum of the
values taken by an affine function over an affine space of dimension at least 2 equals 0. Moreover,  $k$th-order sum-freedom is a strong property and $k$-strong non-normality is a much weaker one (contrary to what  its name evokes). All APN functions (which exist for every $n$) are $k$-strongly non-normal for every $2 \leq k \leq n$ since they are clearly $2$-strongly non-normal and $k$-strongly non-normality is monotonous, i.e., a $k$-strongly non-normal function is $l$-strongly non-normal for every $l\geq k$. In fact, an $(n, m)$-function $F$ is $k$th-order sum-free if, and only if, the restriction of $F$ to any $k$-dimensional affine space, viewed as a $k$-variable function through the choice of a basis of the underlying vector space, has algebraic degree $k$. Therefore, a polynomial $F$ of algebraic degree $d$ is not $l$-th-order sum-free for any $l\geq d+1$ (even more, the summation over all $l$-flats is null). We refer the reader to \cite{CarletTwoNotions2024} for more details on these two properties.

If there exists a $k$-flat $A$ such that the restriction of $F$ to $A$ is affine, then the image of $A$ must be a flat. Thus we have the following proposition.

\begin{prop}\label{prop:break-snn} If $F\colon \F_{2^n}\to \F_{2^n}$ is $k$-breaking, then $F$ is $k$-strongly non-normal.
\end{prop} 
\begin{proof}
    Let $F:\mathbb{F}_{2^n}\to \mathbb{F}_{2^n}$ be a function that breaks all $k$-flats, and suppose that $F$ is not $k$-strongly non-normal. Then, there exists a $k$-flat $A\subseteq \mathbb{F}_{2^n}$ for which the restriction $F|_A$ is affine, i.e., $F(x) = l(x) + b,$ for all $x \in A,$    where $l:\F_{2^n}\to \F_{2^n}$ is a linear function and $b\in \mathbb{F}_{2^n}$.  
    Since $A$ is an affine subspace, we may write it in the form $A = U + a$, where $a \in \mathbb{F}_{2^n}$ and $U$ is a $k$-dimensional linear subspace of $\mathbb{F}_{2^n}$. Then,
    \begin{align*}
        F(A) = \{\, F(u+a) : u \in U \,\} = \{\, l(u+a) + b : u \in U \,\} = \{\, v + c : v \in l(U) \,\} = V + c,
    \end{align*}
    where $V = l(U)$ is a linear subspace of $\mathbb{F}_{2^n}$ and $c = l(a) + b\in\F_{2^n}$.  
    Thus $F(A)$ is an affine subspace of $\mathbb{F}_{2^n}$, contradicting the assumption that $F$ breaks all $k$-flats.
\end{proof}

The converse of Proposition \ref{prop:break-snn} is in general false, as illustrated in the following remark.

\begin{rem}\label{rem:snn not breaking}
Since any APN function is $k$-strongly non-normal, for $k\geq 2$, the multiplicative inverse function $F:\Ftn \to \Ftn$ given by $F(x)=x^{2^n-2}$, for $n$ odd, is $k$-strongly non-normal, for $k\geq 2$. However, it is not $k$-breaking for any divisor $k$ of $n$ since the subfields $\mathbb{F}_{2^k}$ are invariant under $F$, thus unbroken.  Another way to see this is by noting that the property of a function being $k$-strongly non-normal is a CCZ-invariant (Proposition 6 in \cite{CarletTwoNotions2024}), and the breaking property is not EA-invariant (see Remark \ref{rem:breakingCZZ}). 
\end{rem}

The sum-freedom property and the breaking property are two different notions generally. Any APN with at least one preimage of size $3$ has some unbroken $2$-flats by Corollary \ref{th:breakingAPN}. On the other hand, APN functions are always $2$nd-order sum-free. Conversely, the following example exhibits a function which is $3$-breaking, but not $3$rd-order sum-free.

\begin{ex}
    The function $D_2$ from Table \ref{table:unbroken-vanishing-3} breaks all 3-flats, including
    $$ A= \{ 1+\omega+\omega^4+\omega^5, \omega^2, 1+\omega+\omega^2+\omega^3+\omega^5, \omega+\omega^4, 1+\omega^2+\omega^5, 1+\omega^3+\omega^4+\omega^5,\omega+\omega^2+\omega^3,\omega^3+\omega^4\},$$
    where $\omega$ is a generator of $\mathbb{F}_{2^6}^*$ satisfying $\omega^6+\omega^4+\omega^3+\omega+1=0$, for which
    \begin{align*}
    F(A) = \{ & \omega^3+\omega^4+\omega^5,  1+\omega+\omega^2+\omega^4+\omega^5, \omega^4,  1+\omega+\omega^2+\omega^3+\omega^4+\omega^5, \\  &1+\omega+\omega^2, \omega^2+\omega^4, \omega+\omega^2+\omega^3+\omega^4+\omega^5,1+\omega^2+\omega^3\}
    \end{align*}
    is not a flat.  However $\sum_{x\in A}F(x) =0$, so it is not $3$rd-order sum-free.
\end{ex}

For an APN function, a more precise relationship can be derived.

\begin{prop}\label{prop:sumfree_breaking} Let $F:\Ftn\to \Ftn$ be an APN function and $k>2$. If $F$ is $k$-th-order sum-free, then $F$ is $k$-breaking.
\end{prop}

\begin{proof} By Propostion \ref{prop:Sadmir_prop}, if $A\in \mathcal{UB}_{n,k,F}$ then $F(A)$ is a $k$-dimensional flat, thus $A$ is $F$-transversal. Since $\sum_{x\in A} F(x) =\sum_{y \in F(A)} y = 0,$ then $F$ is not $k$-th-order sum-free.
\end{proof}

In general, we have the following result for a $k$-th-order sum-free function.

\begin{prop}\label{prop:sumfree_preimages}  Suppose that $F\colon \F_{2^n}\to \F_{2^n}$ is $k$-th-order sum-free. Let $A$ be a $k$-dimensional flat mapping to an $l$-dimensional flat. Then, there is a positive even number of odd indices $i$ such that $\lambda_{A,i}\not=0$, where $\lambda_{A,i}$ is defined as in \eqref{eq:lambdas_def}, and either $l=1$ or there is at least one even index $j$ such that $\lambda_{A,j}\not=0$. 
\end{prop}

\begin{proof} Let $T_b$ denote the fibers of $F$ w.r.t. $A$, i.e.  $T_b := F^{-1}[b]\cap A= \{ x\in A : F(x)=b\}.$ With this notation, $\sum_{x\in T_b} F(x) = b (|T_b| \pmod{2})$, thus if $|T_b|$ is even,  $\sum_{x\in T_b} F(x) =0$, whereas if $|T_b|$ is odd, $\sum_{x\in T_b} F(x) = b$. Compute the sum $$\sum_{x\in A} F(x) = \sum_{b\in F(A)} \sum_{x\in T_b} F(x) = \sum_{b\in F(A)} b (|T_b| \pmod{2}).$$ 
For the first part, it suffices to see that not all $|T_b|$ can be even, which is a trivial consequence of the fact that the sum $$\sum_{x\in A} F(x) = \sum_{b\in F(A)} b (|T_b| \pmod{2})$$ cannot be null by the sum-freedom of $F$. Suppose that $l\not=1$ (so that $l\geq 2$) and there is no even index $j$ for which $\lambda_{A,j}\not=0$. This implies that all preimages $T_b$ have an odd number of elements, which implies $\sum_{x\in A} F(x) = \sum_{b\in F(A)}b=0$, where the last equality comes from the fact that $F(A)$ is a flat and $|F(A)|> 2$, a contradiction. 
\end{proof}

It is not hard to show \cite{CarletTwoNotions2024} that $F$ is $k$-th-order sum-free if and only if the $k$-th order derivative never takes the zero value. One can be tempted to exploit this to show that if $F$ is $k$-th order sum free, then it must be $(k-1)$-th-order sum-free as the derivative of a zero function is zero. However, the negation of the second part of the statement is: there are $a_1, a_2,\ldots, a_k$ linearly independent elements and $x$ such that $D_{a_1}D_{a_2}\cdots D_{a_k} F(x) =0$. This does not say that the whole derivative is the zero function. In fact, sum-freedom is not a monotonous property. For instance, the multiplicative inverse function, for $n=6$, is not 4th-order sum-free (neither $2$nd nor $3$rd), but it is $5$th-order sum-free. 

\begin{op} Is it true that a $k$-th-order sum-free polynomial that is also $(k-1)$-st-order sum-free is necessarily $k$-breaking, for $k>3$?
\end{op}

\subsection{$k$-strongly breaking functions}\label{sec:strongly_breaking}

For non-bijective APN mappings, there are always vanishing images (e.g.,  when $F$ is 2-to-1, then for every $a \in \F_{2^n}$, there is a unique $b\in \F_{2^n}$ such that $a\not=b$ and $F(a)+F(b)=0$) that may hide additional structural properties. In order to better understand the previously discussed concepts, we introduce the following definition.

\begin{defi}\label{def:sbreaking} Let $F:\F_{2^n}\to \F_{2^n}$ be a function and $A$ be a flat of $\F_{2^n}$. We say that $F$ {\em strongly breaks} $A$ if $\sum_{y \in F(A) }y \not=0,$ where $F(A)$ denotes the image of $A$ under $F$. The function $F$ is said to be $k$-strongly breaking if it strongly breaks all $k$-dimensional flats. 
\end{defi}

Obviously, this notion is equivalent to $k$-th-order sum-freedom for a permutation. Similarly, as in Section \ref{sec:breaking}, we sometimes refer to the strongly breaking property as a property of flats with respect to $F$. Note also that if there is a flat $A$ in the preimage of zero, then $A$ is not strongly broken.

\begin{prop}\label{prop:k-b} Let $F:\F_{2^n}\to \F_{2^n}$ be a function and $A$ be any flat.
\begin{enumerate}[i)]
    \item If $|F(A)|\leq 2$, then $F$ strongly breaks $A$ unless $F(A)=\{0\}$.
    \item If $|F(A)|>2$, then $F$ strongly breaks $A$ $\implies$ $F$ breaks $A$.
\end{enumerate}
\end{prop}

\begin{proof} For $i)$, if the image of $A$ under $F$ contains either a single non-zero point or two different points, then the sum $\sum_{y \in F(A) }y$ is clearly non-null. For $ii)$, note that any set $X\subset \F_{2^n}$ with $|X|>2$, satisfies that $\sum_{x\in X}x \not= 0$ implies that $X$ is not a flat. 
\end{proof}

In general, the converse of $ii)$ in Proposition \ref{prop:k-b} is not true, e.g., it can be verified that Dillon's permutation breaks all 3-flats, see Table \ref{table:unbroken-vanishing-3}, but there are 184 3-flats which are not strongly broken (equivalently, there are 184 vanishing $3$-flats).

Denote by $\mathcal{NSB}_{n,k,F}$ the set of non-strongly broken $k$-dimensional flats w.r.t $F$. Similarly as before, we will typically omit the subindex $k$ when $k=2$.

\begin{prop} Let $F:\F_{2^n}\to \F_{2^n}$ be any function. It holds $$ |\mathcal{UB}_{n,k,F}\setminus\mathcal{C}| = |\mathcal{UB}_{n,k,F}|-|\mathcal{C}|\leq |\mathcal{NSB}_{n,k,F}|,$$
where $\mathcal{C} = \{ A \subset \Ftn : A \mbox{ is a }k\mbox{-dimensional flat, }|F(A)|\leq 2 \}.$ 
\end{prop} 

\begin{proof} Note that $F$ does not break any flat $A$ in $\mathcal{C}$ since all sets with at most two points are flats. This implies that $\mathcal{C} \subset \mathcal{UB}_{n,k,F}$. The result now follows from Proposition \ref{prop:k-b}. 
\end{proof}

\begin{prop} The property of being $k$-strongly breaking is invariant under linear equivalence, i.e., for any linear automorphisms $L$ and $L'$ of $\mathbb{F}_{2^n}$, $F$ is $k$-strongly breaking iff $L'\circ F\circ L$ is $k$-strongly breaking. Furthermore, this is also true if $L$ is an affine automorphism. It holds $|\mathcal{NSB}_{n,k,F} | = |\mathcal{NSB}_{n,k,L'\circ F\circ L} |$ via the correspondence $A \in \mathcal{NSB}_{n,k,F} \mapsto L^{-1}(A)\in  \mathcal{NSB}_{n,k,L'\circ F \circ L} $.
\end{prop}

\begin{proof} Let $L,L'$ and $F$ as in the statement. Suppose that $L'\circ F \circ L$ is not $k$-strongly breaking and let  $A$ be a $k$-flat for which $\sum_{y\in L'\circ F \circ L (A) }y = 0$. Making the substitution $x=(L')^{-1}(y)$ yields $\sum_{x\in F\circ L (A) } L'(x) = 0$. Since $F\circ L( A) = F(L(A))$ and $L'$ is linear, we get $L'(\sum_{x\in F(L (A)) } x) = L'(0)=0$. It must be that $\sum_{x\in F(L (A)) } x =0$ as $L'$ is a permutation. Since $L(A)$ is a $k$-dimensional subspace (or a $k$-flat, when $L$ is affine), we have that $F$ is not $k$-strongly breaking. Conversely, applying the same argument to $L^{-1},(L')^{-1}$ and $F' = L'\circ F \circ L$ yields the result. The last part is now obvious. 
\end{proof}

Using the strongly breaking property, we can deduce a sufficient condition of the APN property.

\begin{lemma}\label{lem:sb-APN}
Let $F:\Ftn \to \Ftn$ be any function. If $F$ is $2$-strongly breaking and there is no $2$-flat $A$ such that, for every $u\in\mathbb{F}_{2^n}$, the number $|F^{-1}[u] \cap A|$ is even, then $F$ is APN.    
\end{lemma}

\begin{proof} Let $A=\{x_1,x_2,x_3,x_4\}$ be a $2$-dimensional flat. We will show that $\sum_{x\in A }F(x) \not=0$. By hypothesis, the only possible cases are the following.

{\bf Case 1}: {\em $F(x_1)=a$, $F(x_2)=b$, and $F(x_3)=F(x_4)=c$, where $a,b,c\in\Ftn$ are pairwise distinct.} In this case, $F(A)=\{a,b,c\}$, so that $\sum_{i=1}^4 F(x_i) = a+b \not= 0$.

{\bf Case 2}: {\em $F(x_1)=F(x_2)=F(x_3)=a$ and $F(x_4)=b$, where $a,b\in\Ftn$ are distinct.} In other words, $F|_A$ is almost $3$-to-$1$ on $A$, and $F(A)=\{a,b\}$. In this case, $\sum_{i=1}^4 F(x_i) = a+b\not= 0$.

{\bf Case 3}: {\em $F(x_1)=a$, $F(x_2)=b$, $F(x_3)=c$, and $F(x_4)=d$, where $a,b,c,d\in\Ftn$ are pairwise distinct. In other words,  $A$ is $F$-transversal.} In this case, $F(A) = \{a,b,c,d\}$, so that $\sum_{i=1}^4 F(x_i) = a+b+c+d = \sum_{y\in F(A)} y\not= 0$ since $F$ is 2-strongly breaking.\\
Therefore, we have shown that $F$ is $2$nd-order sum-free, i.e., $F$ is APN.  
\end{proof}

The converse of Lemma \ref{lem:sb-APN} is in general false, e.g., function $D_3$ in Table \ref{table: APNs unbroken flats} does not strongly break $13$ $2$-flats despite being APN. 

In what follows, we will provide a characterization of a subclass of APN functions in terms of the strongly breaking property. A function $F:\mathbb{F}_{2^n}\to \mathbb{F}_{2^n}$ satisfies property $(\star)$ if for every $A\in \mathcal{B}_{n}$ such that $|F(A)| = 3$, we have \begin{equation}
        F(x)+F(y)+F(z)\not= 0 \tag{$\star$}
    \end{equation} for $x,y,z$ in a subset of $A$ which is $F$-transversal.

\begin{theo}\label{prop:2_strongly_breaking} Let $F:\mathbb{F}_{2^n}\to \mathbb{F}_{2^n}$ be any function for which $F^{-1}[0]$ does not contain any $2$-flat. The following statements are equivalent:
\begin{enumerate}[i)]
    \item $F$ is non-vanishing on any $2$-dimensional $F$-transversal flat and $F$ satisfies $(\star)$.
    \item $F$ strongly breaks all $2$-dimensional flats.
\end{enumerate}
\end{theo}

\begin{proof} Suppose $i)$ holds. Let $A$ be any $2$-flat. By assumption, if $A$ is $F$-transversal or $|F(A)| = 3$, then $F$ strongly breaks $A$. Therefore, we can assume that $|F(A)| \leq 2$. By $i)$ of Proposition \ref{prop:k-b}, we know that these flats are always strongly broken unless $F(A)=\{0\}$ (which cannot happen by assumption). Hence, $F$ strongly breaks all $2$-dimensional flats. Conversely, assume that $ii)$ is true. If $A$ is an $F$-transversal $2$-flat, then $F$ is non-vanishing as it strongly breaks it. Let $A=\{x_1,x_2,x_3,x_4\}$ be any $2$-flat with $|F(A)|=3$, say, $F(x_1) = F(x_2)=a$ and $F(x_3)=b, F(x_4)=c$ for pairwise distinct $a,b,c \in \Ftn$. Since $F$ strongly breaks it, we have $a+b+c \not= 0$ for the transversal subset $\{x_1,x_3,x_4\}$ of $A$. Thus, $F$ satisfies $(\star)$.
\end{proof}

\begin{cor}\label{cor:2stongly} Let $F:\mathbb{F}_{2^n}\to \mathbb{F}_{2^n}$ be any function. The following statements are equivalent:
\begin{enumerate}[i)]
\item $F$ strongly breaks all $2$-flats and there is no $2$-flat $A$ such that, for every $u\in\mathbb{F}_{2^n}$, the number $|F^{-1}[u] \cap A|$ is even.
\item $F$ is an APN function satisfying $(\star)$.
\end{enumerate}
\end{cor}

\begin{proof} The sufficiency follows immediately from Lemma \ref{lem:sb-APN} and Proposition \ref{prop:2_strongly_breaking}. For the necessity, let $F$ be an APN function satisfying $(\star)$. If there were a $2$-flat $A$ such that, for every $u\in\mathbb{F}_{2^n}$, the number $|F^{-1}[u] \cap A|$ is even, then $$\sum_{x\in A} F(x) = \sum_{u\in F(A)} (|F^{-1}[u]\cap A| \pmod{2}) u = 0,$$ contradicting APNness. Thus, such $2$-flats cannot exist. To show that $F$ is $2$-strongly breaking, note that $F^{-1}[0]$ cannot contain a $2$-flat by the above, and, that the function $F$ is non-vanishing on the $F$-transversal $2$-flats since an APN function is non-vanishing on all $2$-flats. By Theorem \ref{prop:2_strongly_breaking}, we infer that $F$ is $2$-strongly breaking.
\end{proof}

\begin{prop} The Gold function $F(x) = x^3$ over $\mathbb{F}_{2^n}$ is $2$-strongly breaking.
\end{prop}

\begin{proof} When $n$ is odd, $F$ is an APN permutation, thus $2$-strongly breaking. Assume $n$ is even. Let $A=\{x,x+a,x+b,x+a+b\}\in \Tf$ be an $F$-transversal 2-dimensional flat. The value $F(x)+F(x+a)+F(x+b)+F(x+a+b)=x^3+(x+a)^3+(x+b)^3+(x+a+b)^3$ equals $a^2b+ab^2=ab(a+b).$ Since $a\not=0,b\not=0$ and $a\not=b$, $F$ is non-vanishing on $A$. Suppose $A=\{x,x+a,x+b,x+a+b\}\in \Tf$ is such that $|F(A)|=3$. W.l.o.g, assume that  $F(x+a)\not= F(x+b)$, $F(x)=F(x+a+b)$, $F(x)\not=F(x+a)$ and $F(x)\not=F(x+b)$. Since $F(x)=F(x+a+b)$ and the preimages of each element form a 2-dimensional space when adjoining the zero element, we have 
\begin{equation}\label{eq:preimages_cube}
x^3 = (x+a+b)^3= (a+b)^3.
\end{equation}
Expanding $(x+a+b)^3$ yields $x^3+(a+b)x^2+ (a^2+b^2)x+(a+b)^3$. By the second equality in Equation \eqref{eq:preimages_cube}, we have 
\begin{equation}\label{eq:zerovalue}
    x^3 + (a+b)x^2+ (a^2+b^2)x =0.
\end{equation} Suppose that $F(x)+F(x+a)+F(x+b)=0,$ equivalently, $x^3+(a+b)x^2+(a^2+b^2)x+(a^3+b^3) =0.$ By Equation \eqref{eq:zerovalue}, it follows $a^3=b^3$, which also implies $(a+b)^3=a^3=b^3$. Equation \eqref{eq:preimages_cube} tells us then that either $x=a$ or $x=b$, which contradicts $F(x+b)\not=F(x)$ or $F(x+a)\not=F(x)$, respectively. By Theorem \ref{prop:2_strongly_breaking}, we conclude that $F(x)$ is $2$-strongly breaking.
\end{proof}

\subsection{Relationship among all notions}\label{sec:relationship}

The following diagrams are visualizations of the relationship between the introduced concepts.

\begin{figure}[!htb]
    \centering
 \begin{tikzpicture}[
    every node/.style={font=\normalem},
    textblock/.style={
        draw,
        rounded corners,
        inner sep=3pt
    }
]

% Top
\node[textblock] (strongbreaking) at (0,0)
    {\(F\) is \(k\)-strongly breaking};

% Arrow
\node at (0,-0.65) {$\Downarrow$%
\asterisk};

% Second
\node[textblock] (breaking) at (0,-1.25)
    {\(F\) is \(k\)-breaking};

% Arrow
\node at (0,-1.90) {$\Downarrow$ $ $};

% Middle row
\node[textblock] (sumfree) at (-5.5,-2.55)
    {\(F\) is \(k\)th-order sum-free};

\node (implies) at (-2.9,-2.55)
    {$\Rightarrow$};

\node[textblock] (nonnormal) at (0,-2.55)
    {\(F\) is \(k\)-strongly non-normal};

% Arrow
\node (otherimplies) at (2.9,-2.55) 
 {$\Leftarrow$};
%{$\Downarrow\;\text{\scriptsize $(k=2)$}$};

% Bottom
\node[textblock] (apn) at (4.25,-2.55)
    {\(F\) is APN};

\end{tikzpicture}
\caption{Relationship between the introduced concepts, where * means that $F$ is such that $|F(A)|>2$ for any $k$-flat $A$.}
\label{fig:rel}
\end{figure}
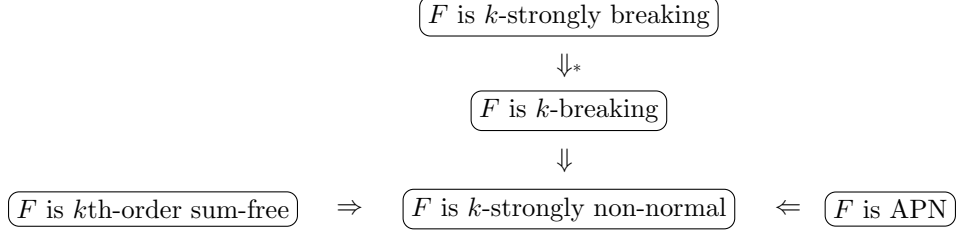

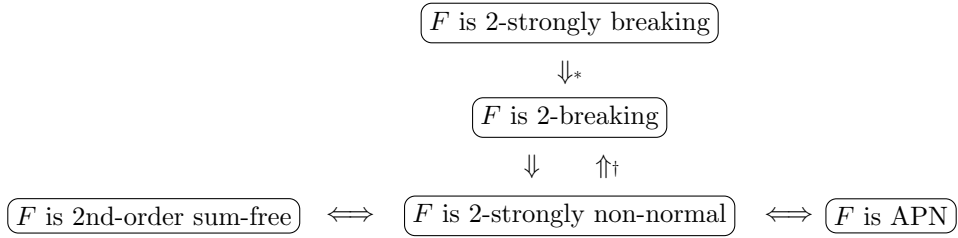
\begin{figure}[!htb]
    \centering
 \begin{tikzpicture}[
    every node/.style={font=\normalem},
    textblock/.style={
        draw,
        rounded corners,
        inner sep=3pt
    }
]

% Top
\node[textblock] (strongbreaking) at (0,0)
    {\(F\) is $2$-strongly breaking};

% Arrow
\node at (0,-0.65) {$\Downarrow$%
\asterisk
}
;

% Second
\node[textblock] (breaking) at (0,-1.25)
    {\(F\) is $2$-breaking};

% Arrow
\node at (-0.5,-1.90) {$\Downarrow$};

\node at (0.5,-1.90) {$\Uparrow$%
\dagg};

% Middle row
\node[textblock] (sumfree) at (-5.5,-2.55)
    {\(F\) is $2$nd-order sum-free};

\node (implies) at (-2.9,-2.55)
    {$\iff$};

\node[textblock] (nonnormal) at (0,-2.55)
    {\(F\) is $2$-strongly non-normal};

% Arrow
\node (otherimplies) at (2.9,-2.55) 
 {$\iff$};
%{$\Downarrow\;\text{\scriptsize $(k=2)$}$};

% Bottom
\node[textblock] (apn) at (4.25,-2.55)
    {\(F\) is APN};

\end{tikzpicture}
\caption{Relationship between the introduced concepts for $k=2$, where * means that $F$ is such that $|F(A)|>2$ for any $2$-flat and \dag\  indicates that $F$ is $2$-breaking on all flats except on those for which it is almost three-to-one.}
\label{fig:rel2}
\end{figure}

\section{Balanced properties related to known APN functions}\label{sec:balanced}

In this section, we study certain balanced properties of APN functions (and related functions) by using the structure of the collection of $2$-flats and exploiting Fourier transformations on this set. A direct relationship between these transformations and the sum-of-square indicator of vectorial functions is derived and then used to analyze certain classes of APN functions. 

Consider the set of all 2-dimensional flats of $\Ftn$, denoted by $\mathcal{F}_{2,n}$. It is  well-known that $|\mathcal{F}_{2,n}| = \frac{2^{n-2}(2^{n-1}-1)(2^n-1)}{3}$ \cite{BraWoPre05}.  Given a function $F:\Ftn\to \Ftn,$ the mapping $\Phi_F: \mathcal{F}_{2,n} \to \F_{2^n}$\footnote{Note that the function $\Phi_F$ was defined by Carlet in the context of AB functions, Dillon pointed out its surjectiveness, when $F$ is APN \cite{CarletBook}.} is defined by \begin{equation}\label{eq:phi}
    \Phi_F(A)=\sum_{x \in A} F(x).
\end{equation} 
It is easy to see that a function $F$ is APN if and only if the image of $\Phi_F$ is $\mathbb{F}_{2^n}^*$.  For each $v\in \Ftn^*$, we consider the following character sum of the components of $\Phi_F$,
\begin{equation}\label{eq:sv}
    S_{v,F} = \sum_{A\in \mathcal{F}_{2,n}} (-1)^{\operatorname{Tr}_n(v \Phi_F(A))}.
\end{equation}
Observe that $$\sum_{v\in \Ftn} S_{v,F} = \sum_{v\in \Ftn} \sum_{A\in \mathcal{F}_{2,n}} (-1)^{\operatorname{Tr}_n(v\Phi_F(A))} = \sum_{A\in \mathcal{F}_{2,n}} \sum_{v\in \Ftn} (-1)^{\operatorname{Tr}_n(v\Phi_F(A))}.$$ The sum $\sum_{v\in \Ftn} (-1)^{\operatorname{Tr}_n(v\Phi_F(A))}$ equals $0$ if $\Phi_F(A)\not = 0$. Otherwise, $\sum_{v\in \Ftn} (-1)^{\operatorname{Tr}_n(v\Phi_F(A))}=2^n$. Note that $\mathcal{Z}_{n,F} = \{A\in \mathcal{F}_{2,n} : \Phi_F(A) =0\}$. Thus,
\begin{equation}\label{eq:zeros_of_phi}
    |\mathcal{Z}_{n,F}|  = |\{A\in \mathcal{F}_{2,n} : \Phi_F(A) =0\}| = \frac{1}{2^n}\sum_{v\in \Ftn} S_{v,F} = \frac{\sum_{v\in \Ftn^*} S_{v,F} + |\mathcal{F}_{2,n}|}{2^n}.
\end{equation}
Define $\mathcal{NZ}_{n,F}:=\{A\in \mathcal{F}_{2,n}\colon \Phi_F(A) \neq 0\}$, so that $\mathcal{F}_{2,n}=\mathcal{NZ}_{n,F}\cup \mathcal{Z}_{n,F}$ is a partition of the set of all 2-flats.

\begin{lemma}\label{lem:Sv} Let $F:\Ftn \to \Ftn$ be any function and $\Phi_F: \mathcal{F}_{2,n} \to \F_{2^n}$ be its associated function, defined by \eqref{eq:phi}. For each nonzero $v\in \Ftn$, consider the sum $S_{v,F}$ defined in \eqref{eq:sv}. It holds:
\begin{enumerate}[i)]
    \item $F$ is APN if and only if $\sum_{v\in \Ftn^*} S_{v,F} = -|\mathcal{F}_{2,n}|$.
    \item For each nonzero $v\in \Ftn$, the value of $S_{v,F}$ is directly related to the sum-of-square indicator of the component $F_v$. Namely, $$ S_{v,F}  = \frac{1}{24} (\nu(F_v) - 2^{2n+1} + 2^{n+1} - 2^{2n}).$$
    
\item $\sum_{v\in \Ftn^* } \nu(F_v) = 3\cdot 2^{n+3} |\mathcal{Z}_{n,F}| +(2^n-1)2^{2n+1}.$
\end{enumerate}
\end{lemma} 
\begin{proof} The first item is trivial from \eqref{eq:zeros_of_phi}. Let us show $ii)$. Since each flat can be represented by $3!\binom{4}{3} = 24$ triplets with pairwise distinct elements, we have 
 $$S_{v,F}=\sum_{A\in \mathcal{F}_{2,n}} (-1)^{\operatorname{Tr}_n(v \Phi_F(A))} =\frac{1}{24}\sum_{\substack{(x,y,a)\in \Ftn^2\times \Ftn^*,\\ x\not=y, x\not=y+a}} (-1)^{F_v(x)+F_v(x+a)+F_v(y)+F_v(y+a)}.$$  This equality can then be rewritten as 
 
 \begin{align*}
 &\frac{1}{24}\sum_{\substack{(x,y,a) \in \Ftn^2\times \Ftn^*,\\ x\not=y, x\not=y+a}} (-1)^{F_v(x)+F_v(x+a)+F_v(y)+F_v(y+a)}=\sum_{\substack{(x,y,a)\in \Ftn^2\times \Ftn^*,\\ x\not=y, x\not=y+a}} (-1)^{D_aF_v(x)+D_aF_v(y)},
 \end{align*}
 where $v\in\Ftn^*$.
On the other hand, for nonzero $v\in\Ftn$, $$\sum_{a\in\Ftn^*}W_{D_aF_v}^2(0) = \sum_{a\in\Ftn^*}\left(\sum_{x\in \Ftn} (-1)^{ D_a F_v(x)}\right)^2 =
\sum_{(x,y,a)\in \Ftn^2\times \Ftn^*} (-1)^{D_aF_v(x)+D_aF_v(y)}.$$
Thus, $$S_{v,F}=\frac{1}{24}\left( \sum_{a\in\Ftn^*}W_{D_aF_v}^2(0) - 2^{n+1}(2^n-1)\right),$$ where the number $2^{n+1}(2^n-1)$ comes from the number of \emph{forbidden} values given by 
$$\sum_{a\in\Ftn^*} \sum_{x\in \Ftn^2} \sum_{y=x, y= x+a} (-1)^{{D_aF_v(x)}+D_aF_v(y)}.$$
Then,  
$ S_{v,F} = \frac{1}{24} \left(\nu(F_v) - 2^{2n}-2^{2n+1}+2^{n+1} \right),$ where we tacitly used $W_{D_0F_v}^2(0) = 2^{2n}.$  Item $iii)$ follows now from $ii)$ and \eqref{eq:zeros_of_phi}.
\end{proof}

\begin{rem}\label{rem:solution_OP4} Note that $iii)$ in Lemma \ref{lem:Sv} generalizes Corollary 1 of \cite{BeCaChLa06} thereby giving the exact value of $\sum_{v\in \Ftn^* } \nu(F_v)$. Moreover, we provide a positive answer to Open Problem 4 in \cite{BeCaChLa06}, which asks for the existence of a 4-uniform permutation over $\Ftn$, for $n$ even, such that $$\sum_{v\in \Ftn^* } \nu(F_v) = (2^n-1)2^{2n+1} + 3\cdot 2^{n+3} |\mathcal{Z}_{n,F}| := (2^n-1)2^{2n+1} + A2^{n+3},$$ and $A=3 |\mathcal{Z}_{n,F}| < 2^n-1.$  In general, for a $4$-uniform function $F$, $$3|\mathcal{Z}_{n,F}| = \sum_{a\in\Ftn^*, b\in \Ftn} \binom{\delta_F(a,b)/2}{2} = |\{ (a,b) \in \Ftn^*\times \Ftn : \delta_F(a,b) = 4 \}|.$$
Thus, the problem is equivalent to finding a $4$-uniform permutation with strictly less than $2^n-1$ fours in its differential distribution table. For this, consider Dillon's APN permutation $G$ over $\mathbb{F}_{2^6}$, where $\omega$ generates $\F_{2^6}^*$ and it annihilates the polynomial $x^6 + x^4 + x^3 + x + 1$. Define $F$ as follows, 
$$F(x) = 
\begin{cases}
G(x), & x\not\in \{\omega, \omega^{-2}\};\\
G(\omega), & x = \omega^{-2};\\
G(\omega^{-2}), & x = \omega.
\end{cases}$$
In other words, $F$ is $G$ precomposed by the transposition $(\omega\ \omega^{-2})$ (swapping two values). It can be verified that $F$ is a $4$-uniform permutation with $54<2^6-1=63$ fours in its difference distribution table.   
\end{rem}

We are interested in determining when $\Phi_F$ is balanced over $\Ftn^*$.  Since the codomain of $\Phi_F$ is $\Ftn$, there are two possibilities according to whether $\Phi_F^{-1}(0) = \emptyset$ or not. We will say that $\Phi_F$ is \emph{balanced} if for any $x,x'\in\mathbb{F}_{2^n}^*$, $|\Phi_F^{-1}[x]|=|\Phi_F^{-1}[x']|,$
 and $\Phi_F^{-1}[0] = \emptyset$. On the other hand, we will call $\Phi_F$ \emph{balanced-like} if for any $x,x'\in\mathbb{F}_{2^n}^*$, $|\Phi_F^{-1}[x]|=|\Phi_F^{-1}[x']|>0,$
 and $\Phi_F^{-1}[0]\neq \emptyset$. 
If $\Phi_F$ is balanced, then, for each nonzero $x\in\Ftn$, $$|\Phi_F^{-1}[x]| = \frac{|\mathcal{F}_{2,n}|}{2^n-1} = \frac{2^{2n-3}-2^{n-2}}{3}.$$ Since $2^{2n-3}-2^{n-2} \equiv (-1)^{2n-3} - (-1)^{n-2} \pmod{3}$, it must be that $2n-3\equiv n-2 \pmod{2}$, equivalently, $n$ must be odd. If $\Phi_F$ is balanced-like, it holds that $$|\Phi_F^{-1}[x]| = \frac{|\mathcal{F}_{2,n}|- |\mathcal{Z}_{n,F}|}{2^n-1} = \frac{2^{3n-3} -2^{2n-2}-2^{2n-3} +2^{n-2} - \sum_{a\in\Ftn^*, b\in \Ftn}\binom{\delta(a,b)}{2}}{3(2^n-1)}.$$

To analyze the balancednees properties of $\Phi_F$, define the function $\psi_F:\Ftn\to\mathbb{C}$ given by \begin{equation}\label{eq:psi}
    \psi_F(c) = |\Phi_F^{-1}[c]|
\end{equation} The Fourier transform of $\psi$ is $$\widehat{\psi_F}(v) = \sum_{c\not=0} \psi_F(c) \chi_{v}(c),$$ where $\chi_v:\Ftn\to \mathbb{C}$ denotes the character $\chi_v(c) = (-1)^{\operatorname{Tr}_n(c v)}$. Since $\widehat{\psi_F}(v) = \sum_{A\in \mathcal{F}_{2,n}} \chi_{v}(\Phi_F(A))$,  we get  $\widehat{\psi_F}(v) = \sum_{A\in \mathcal{F}_{2,n}}(-1)^{\operatorname{Tr}_n(v \Phi_F(A))} = S_{v,F}.$
Using Fourier inversion on $\Ftn$, 
\begin{equation}\label{eq:Fourier}
\psi_F(c) =2^{-n} \sum_{v\in \Ftn} \widehat{\psi_F}(v)\chi_v(c).
\end{equation}

\begin{lemma} \label{lemma:characterization_balanced} 
Let $n$ be an odd integer, $n>1$. The function $\Phi_F: \mathcal{F}_{2,n} \to \F_{2^n}$ is balanced if and only if for every nonzero $v\in \Ftn$, the number $S_{v,F}$, as defined in \eqref{eq:sv}, is independent of $v$ and $\Phi_F^{-1}[0] = \emptyset$ (i.e., $F$ is APN), in which case $S_{v,F} = -\frac{|\mathcal{F}_{2,n}|}{2^n-1}.$
\end{lemma}
\begin{proof}
Suppose that $\Phi_F$ is balanced. The equality 
\begin{align*}
   S_{v,F}= \sum_{A\in \mathcal{F}_{2,n}} (-1)^{\operatorname{Tr}_n(v \Phi_F(A))} &=  \sum_{c\in\Ftn^*} |\Phi^{-1}[c]|(-1)^{\operatorname{Tr}_n(v c)}  = \frac{|\mathcal{F}_{2,n}|}{2^n-1}\sum_{c\in\Ftn^*}(-1)^{\operatorname{Tr}_n(v c)} = -\frac{ |\mathcal{F}_{2,n}|}{2^n-1}.
\end{align*}
implies that $S_{v,F}$ is independent of $v$. The sufficiency follows now immediately from \eqref{eq:zeros_of_phi} since it gives $\Phi_F^{-1}[0] = \emptyset$. Suppose that for each $v\in \Ftn^*$, $S_{v,F} = S$ is independent of $v$ and $\Phi_F^{-1}[0] = \emptyset$. Consider the function $\psi_F$ as defined in \eqref{eq:psi}. Since $\widehat{\psi_F}(0)=|\mathcal{F}_{2,n}|$ and $\widehat{\psi_F}(v) = S$, we get 
    $\psi_F(c) = 2^{-n} \left( |\mathcal{F}_{2,n}|+S\sum_{v\not=0}\chi_v(c)\right).$
For $c\not=0$, this gives $\psi_F(c) = 2^{-n} \left( |\mathcal{F}_{2,n}|- S\right)$. Since  $\psi_F(0) = 0,$ it must be $S = -\frac{|\mathcal{F}_{2,n}|}{2^n-1}$ by \eqref{eq:zeros_of_phi}. Therefore, $\psi_F(c) =-S$. 
\end{proof}

\begin{lemma} \label{lemma:characterization_balanced-like} The function $\Phi_F: \mathcal{F}_{2,n} \to \F_{2^n}$ is balanced-like if and only if $\Phi_F^{-1}[0] \not= \emptyset$ and, for every nonzero $v\in \Ftn$, the sum $S_{v,F}$, defined in \eqref{eq:sv}, is independent of $v$, in which case $S_{v,F} = \frac{ 2^n|\mathcal{Z}_{n,F}|-|\mathcal{F}_{2,n}|}{2^n-1}.$
\end{lemma}

\begin{proof}
The number $S_{v,F}=\sum_{A\in \mathcal{F}_{2,n}} (-1)^{\operatorname{Tr}_n(v \Phi_F(A))}$ is equal to 
\begin{align*}
    |\Phi_F^{-1}[0]|+ \sum_{c\in\Ftn^*} |\Phi^{-1}[c]|(-1)^{\operatorname{Tr}_n(v c)} = |\mathcal{Z}_{n,F}| + \frac{|\mathcal{NZ}_{n,F}|}{2^n-1}\cdot \sum_{c\in\Ftn^*}(-1)^{\operatorname{Tr}_n(v c)} = \frac{2^n|\mathcal{Z}_{n,F}| - |\mathcal{F}_{2,n}|}{2^n-1}.
\end{align*}
Then, the sufficiency follows. For the necessity, suppose that for each $v\in \Ftn^*$, $S_{v,F} = S$ is independent of $v$. We proceed similarly as in the proof of Lemma \ref{lemma:characterization_balanced} and obtain that in this case the Fourier transform of $\psi_F$ is $$\widehat{\psi_F}(v) = \psi_F(0) \chi_v(0) +\sum_{c\neq 0} \psi_F(c) \chi_{v}(c) = \psi_F(0) +\sum_{c\neq 0} \psi_F(c) \chi_{v}(c). $$ As $\widehat{\psi_F}(v) = \sum_{A\in \mathcal{F}_{2,n}}(-1)^{\operatorname{Tr}_n(v \Phi_F(A))}.$
Then, for $c\not=0$,
we have, $\psi_F(c) = 2^{-n} \left( |\mathcal{F}_{2,n}|-S\right) = S,$ and for $c=0$, it gives $\psi_F(c) = 2^{-n}(|\mathcal{F}_{2,n}|+ S\cdot (2^n-1))= \frac{S-|\mathcal{F}_{2,n}|}{2^n}$=$|\mathcal{Z}_{n,F}|$.
\end{proof}

In fact, the balanced property of $\Phi_F$ is equivalent to the CAPNness of $F$. A proof of this fact was essentially given in \cite{CarletCAPN18}. We provide a direct argument using our context. 

\begin{theo}\label{th:CAPN} Let $n>1$ be an odd integer (resp. any integer) and $F:\Ftn\to\Ftn$ be any function. The following properties are equivalent.

\begin{enumerate}[i)]
    \item The function $\Phi_F$ is balanced (resp. balanced-like);
    \item For every nonzero $v\in\Ftn$, $\nu(F_v)$ takes on a single value, in which case $\nu(F_v)=2^{2n+1}$ (resp. $\nu(F_v) = 2^{2n+1} +  \frac{3\cdot 2^{n+3}|\mathcal{Z}_{n,F}|}{2^n-1}$);
    \item For every nonzero $v\in\Ftn$, we have $\sum_{a\in\Ftn} W_{F_v}^4(a)=2^{3n+1}$, that is, $F$ is CAPN (resp. $\sum_{a\in\Ftn} W_{F_v}^4(a)=2^{3n+1} + \frac{3\cdot 2^{2n+3}|\mathcal{Z}_{n,F}|}{2^n-1}$).
\end{enumerate}
\end{theo}

\begin{proof} The equivalence $i) \iff ii)$ follows from Lemmata \ref{lem:Sv} and  \ref{lemma:characterization_balanced}. Namely, $\Phi_F$ is balanced if and only if, for each $v\in\Ftn^*$, $S_{v,F} =-\frac{|\mathcal{F}_{2,n}|}{2^n-1} = \frac{2^{n-2}-2^{2n-3}}{3}$ if and only if, for each $v\in\Ftn^*$, $\nu(F_v)=2^{2n+1}.$
Similarly, we can obtain the respective result for the balanced-like property by applying Lemma \ref{lemma:characterization_balanced-like} instead of Lemma \ref{lemma:characterization_balanced}.

Suppose that $ii)$ holds. If $\nu(F_v)$ takes on a single value for each $v\in \Ftn^*$, then so does $S_{v,F}$ by the equivalence $i) \iff ii)$, say, that this value is $S$. From \eqref{eq:zeros_of_phi}, it must be that $2^n |\mathcal{Z}_{n,F}|-|\mathcal{F}_{2,n}|=\sum_{v\not=0} S_{v,F} = S (2^n-1)$, so that $S = \frac{2^n |\mathcal{Z}_{n,F}|-|\mathcal{F}_{2,n}|}{2^n-1}$. Thus, $$\nu(F_v) = 24S + 2^{2n+1} - 2^{n+1} + 2^{2n} = \frac{3\cdot 2^{n+3} |\mathcal{Z}_{n,F}|}{2^n-1}+2^{2n+1},$$
by Lemma \ref{lem:Sv}. The equivalence $ii) \iff iii)$ now follows from the equation $\nu(F_v)=2^{-n}\sum_{a\in\Ftn} W_{F_v}^4(a).$
\end{proof}

From Theorem \ref{th:CAPN}, we note that the balanced-like property of $\Phi_F$ gives a natural extension of the CAPN concept to all dimensions.    

\begin{defi} The function $F:\Ftn\to \Ftn$ is said to be \emph{CAPN-like} if $\Phi_F$ is balanced-like.
\end{defi}

\begin{ex}\label{ex:power_perm} Let $F(x)=x^d$ be any power permutation defined over $\mathbb{F}_{2^n}$, where $\gcd(2^n-1,d)=1$. By Proposition 5 in \cite{BeCaChLa06}, we know that $\nu(F_v) = \nu(F_1) = \sum_{\lambda\in \Ftn} W_{D_1 F_{\lambda}}^2(0)$ for each $v\in\Ftn^*$, i.e., $\nu(F_v)$ takes on a single for each $v\in\Ftn$, thus $$\nu(F_v) = 2^{2n+1} +  \frac{3\cdot 2^{n+3}|\mathcal{Z}_{n,F}|}{2^n-1}$$ and $F$ is CAPN or CAPN-like, depending on wheter $F$ is APN ($n$  odd, necessarily), or not by Theorem \ref{th:CAPN}. Note that there are CAPN-like permutations for $n$ odd, too. For instance, the function $F(x)=x^7$ over $\F_{2^7}$ is a $6$-uniform permutation, thus it is CAPN-like.
\end{ex}

\begin{ex} Continuing with Example \ref{ex:power_perm}, the multiplicative inverse function $F$ given by $x\mapsto x^{2^n-2}$ is CAPN for $n$ odd and CAPN-like for $n$ even. Indeed, for the multiplicative inverse permutation, the values of the  sum-of-square indicator are given by $$\nu(F_v) = \begin{cases}
    2^{2n+1}, & 2\nmid n,\\
    2^{2n+1} +  2^{n+3}, & 2\mid n,
\end{cases}$$
since $F$ is APN for $n$ odd, whereas for $n$ even, it is $4$-differentially uniform, and $$|\mathcal{Z}_{n,F}|= \frac{1}{3}\sum_{a\in\Ftn^*,\ b\in \Ftn}\binom{ \delta(a,b)/2}{2}=\frac{2^n-1}{3}.$$ 
Theorem \ref{th:CAPN} implies that $\Phi_F$ is balanced (resp. balanced-like), when $n$ is odd (resp. $n$ is even).
\end{ex}

We say that $\Phi_F$ is $k$-\emph{balanced} if we can partition $\mathcal{F}_{2,n} = \cup_{i=1}^k \mathcal{F}_i$, where $ \mathcal{F}_i \cap \mathcal{F}_{j}=\emptyset$, for each $1\leq i<j\leq k$ such that $\Phi_F|_{\mathcal{F}_{i}}$ is balanced, for each $1\leq i\leq k$. In other words, $\Phi_F$ is $k$-balanced if for any $x\in \mathbb{F}_{2^n}$ we have
$$|\Phi_F^{-1}[x]|= \begin{cases}
  \frac{|\mathcal{F}_{i}|}{|V_i|},  & \text{if }x\in V_i, \\
\end{cases}$$
where $\Ftn^*=\cup_{i=1}^k V_i$, $V_i\cap V_{j}=\emptyset$, for each $1\leq i<j\leq k$. Note that this definition implies that $F$ is APN. 
\begin{defi} A function $F:\Ftn \to \Ftn$ is $k$-CAPN if $\Phi_F$ is $k$-balanced.
\end{defi}

For any (APN) function $F:\Ftn \to \Ftn$, there always exists $1\leq k\leq 2^n $ such that  $F$ is $k$-CAPN. If $k=1$, this is just the balanced definition of $\Phi_F$. We will particularly be interested in the case when $k=2.$ In the following, an affine function $L:\Ftn\to \Ftn$ will be represented by a linear function $M:\Ftn\to\Ftn$ and an element $b\in\Ftn$ in such a way that $L(x) = M(x)+b$. For different affine functions, the underlying linear functions and elements will be clear from the context.

\begin{prop} Let $F:\Ftn \to \Ftn$ be any function and let $F':\Ftn\to \Ftn$ be a function EA-equivalent to $F$, i.e., there exist affine automorphisms $L, L'$ of $\Ftn$ and an affine function $L'':\Ftn \to \Ftn$ such that $F'=L\circ F \circ L'+L''$. Then, $\Phi_{L\circ F \circ L'+L''}(A) = M(\Phi_{F}( L'(A) )).$ In particular, the CAPN, CAPN-like and the $k$-CAPN properties are EA-invariant.
\end{prop}

\begin{proof} Compute the value for $\Phi_{L\circ F \circ L' + L''}(A)$, namely, 
$$\Phi_{L\circ F \circ L' + L''}(A) = \sum_{x\in A} \big( L\circ F \circ L'(x) + L''(x) \big) = \sum_{y\in L'(A) } L(F(y)) = M\big(\sum_{y\in L'(A) } F(y) \big).$$ Since $L'(A)$ is a $2$-flat, the value of $\Phi_F( L'(A))$ is well-defined, and $\Phi_{L\circ F \circ L' + L''}(A) = M(\Phi_F( L'(A)))$. As $M$ is an automorphism, it holds that, for any $c\in\Ftn$, $\Phi_F(A) = c$ if and only if $M(\Phi_F(A)) = M(c)$. Since $L'$ is an affine automorphism, we have that $M(\Phi_F(A)) = M(c)$ if and only if $M(\Phi_F( L'(B) )) = M(c)$ for a unique $B\in \mathcal{F}_{2,n}$ such that $L'(B)=A$. By the first part, $M(\Phi_F( L'(B) )) =  \Phi_{L\circ F\circ L'+L''}(B).$ Putting all together yields, for each $c\in\Ftn$, $$\Phi_F(A) = c \iff  \Phi_{L\circ F\circ L'+L''}(B) = M(c)$$ for $B=(L')^{-1}(A)$. This equivalence implies that $|\Phi_F^{-1}[c]|  = |\Phi_{L \circ F\circ L'+L''}^{-1} [M(c)]|$ for any $c\in\Ftn$. From here, it is clear that if $F$ is balanced, balanced-like or $k$-balanced, then so is $F'$.
\end{proof}

\begin{rem} In general, the property of being $k$-CAPN is not CCZ-invariant. For instance, it was shown in \cite{Budaghyan2006} that the function $F(x)=x^3+(x^2+x+1)\operatorname{Tr}_n(x^3)$ defined over $\mathbb{F}_{2^6}$ is CCZ-equivalent EA-inequivalent to $x^3$. However, computer-based simulations show that $F$ is a $3$-CAPN function, whereas $x^3$ is CAPN.
\end{rem}

\begin{lemma} \label{lemma:characterization_even} Let $\Phi_F: \mathcal{F}_{2,n} \to \F_{2^n}$ be a $2$-balanced function, whose associated partition of $\Ftn^*$ is given by $(V_1,V_2)$, such that $\gamma_i := \frac{|\Phi_F^{-1}[V_i]|}{|V_i|}$ for $i\in\{1,2\}$. Suppose that one (exactly one, necessarily) of $V_1\cup\{0\}$ or $V_2\cup\{0\}$ is an $l$-dimensional subspace, denoted by $W$. Then, $S_{v,F}$, as defined in \eqref{eq:sv}, is $2$-valued. Moreover, 
$$ S_{v,F} = 
\begin{cases}
   -\gamma_i, & v\not\in W^\perp;\\
   2^l(\gamma_i-\gamma_{j}) - \gamma_i, &  v \in W^\perp\setminus \{0\},\\
\end{cases}
$$
where $i,j$ are such that $W = V_i\cup\{0\}$ and $V_j = \Ftn \setminus W$, respectively, and $\perp$ denotes the orthogonal complement. Conversely, if $F$ is APN and, for every $v\in \Ftn^*$, $S_{v,F}\in \{\theta_1 ,\theta_2\}$ with $\theta_1\not=\theta_2$ such that either $\{v\in \Ftn : S_{v,F} = \theta_1\}\cup\{0\}$ or $\{v\in \Ftn : S_{v,F} = \theta_2\}\cup\{0\}$ is a subspace. Then $\Phi_F$ is $2$-balanced.
\end{lemma}
\begin{proof} Without loss of generality, we will assume that $W := V_2\cup \{0\}$ is a subspace. Let $v\in \Ftn^*$ and consider the sum $S_{v,F} = \sum_{A\in \mathcal{F}_i} (-1)^{\operatorname{Tr}_n(v\Phi_F(A))}$. Since $\Phi_F^{-1}[0] = \emptyset$, we have 
\begin{equation}\label{eq:Sv_2valued}
    S_{v,F} = \sum_{c\in \Ftn^*} |\Phi_F^{-1}[c]|(-1)^{\operatorname{Tr}_n(vc)} = \gamma_1 \sum_{c\in V_1} (-1)^{\operatorname{Tr}_n(vc)} +  \gamma_2 \sum_{c\in V_2} (-1)^{\operatorname{Tr}_n(vc)}.
\end{equation}
Since $\sum_{c\in V_1} (-1)^{\operatorname{Tr}_n(vc)} + \sum_{c\in V_2} (-1)^{\operatorname{Tr}_n(vc)} + 1 = \sum_{c\in \Ftn} (-1)^{\operatorname{Tr}_n(cv)} = 0$, the value of the sum over $V_1$ is fully determined by the value of the sum over $V_2$. By the \emph{Poisson summation formula} (see \cite{CarletBook}), we get $$ \sum_{c\in V_2} (-1)^{\operatorname{Tr}_n(vc)} +1 = \frac{1}{|W^\perp|}\sum_{u\in W^\perp } \sum_{x\in\Ftn} \chi_{v+u}(x).$$
The sum $\sum_{x\in\Ftn} \chi_{v+u}(x)$ is null unless $v=u$, thus, $$\frac{1}{|W^\perp|}\sum_{u\in W^\perp } \sum_{x\in\Ftn} \chi_{v+u}(x) = \begin{cases}
    \frac{2^n}{|W^\perp|},& v\in W^{\perp}\\
    0,& v\not\in W^{\perp}.
\end{cases}  $$
From here, we obtain $$S_{v,F} = \frac{2^n}{|W^\perp|}(\gamma_2-\gamma_1) -\gamma_1 = 2^l (\gamma_2-\gamma_1) -\gamma_2,$$ whenever $v\in W^\perp\setminus\{0\}$, and $S_{v,F} = -\gamma_2$, otherwise.

Conversely, assume that $S_{v,F}\in \{\theta_1,\theta_2\}$. W.l.o.g, suppose that $U := \{v\in \Ftn : S_{v,F} = \theta_2\}\cup\{0\}$ is a subspace. Consider $\psi_F$ as defined in \eqref{eq:psi}. From \eqref{eq:Fourier}, we get 
\begin{equation}\label{eq:psicnew}
    \psi_F(c) = 2^{-n} \left( |\mathcal{F}_{2,n}|+\theta_2\sum_{v\in U\setminus \{0\}}\chi_v(c) + \theta_1\sum_{v\in \Ftn \setminus U}\chi_v(c)\right)
\end{equation}
since $\widehat{\psi_F}(0)=|\mathcal{F}_{2,n}|$
By another application of the Poisson summation formula, we get $$\sum_{v\in U\setminus \{0\}}\chi_v(c) + 1 =  \frac{1}{|U^\perp|}\sum_{u\in U^\perp } \sum_{x\in\Ftn} \chi_{c+u}(x) = \begin{cases}
    \frac{2^n}{|U^\perp|},& c\in U^{\perp}\\
    0,& c\not\in U^{\perp}.
\end{cases}$$
Note that $\sum_{v\in \Ftn \setminus U}\chi_v(c) = - \sum_{v\in U\setminus\{0\}}\chi_v(c) -1$. From \eqref{eq:psicnew}, it must be that, for $c\not=0$, $$\psi_F(c) = \begin{cases}
    \frac{|\mathcal{F}_{2,n}|}{2^n} +\frac{(\theta_2-\theta_1)}{|U^\perp|} - \frac{\theta_2}{2^n},& c\in U^{\perp}\setminus\{0\}\\
    \frac{|\mathcal{F}_{2,n}|}{2^n} -\frac{\theta_2}{2^n},& c\not\in U^{\perp}.
\end{cases}$$
In other words, $\Phi_F$ is $2$-balanced.
\end{proof}

Let $\alpha$ be a primitive element of $\Ftn^*$. Define cyclotomic classes of order 3 
$$
C_j:=\{\alpha^{3k+j}\mid 0\leq k<\tfrac{2^n-1}{3}\},
\qquad j\in\{0,1,2\}.
$$
It is easy to see that $vC_j = vC_{(j+i) \mod 3}$ when $v\in C_i$, for $i,j\in \{0,1,2\}$. Define
$$
\eta_j:=\sum_{c\in C_j}(-1)^{\operatorname{Tr}_n(c)},
\qquad j\in\{0,1,2\}.
$$
so that $\eta_0+\eta_1+\eta_2=-1.$  Moreover, when $n$ is even, the classes $C_1$ and $C_2$ are nonempty. Since the Frobenius automorphism restricted to $C_1$ has image $C_2$, $|C_1|=|C_2|$. 
It then holds that $\eta_1 = \eta_2$ as $\operatorname{Tr}_n(\alpha^{3k+1})=0 \implies \operatorname{Tr}_n(\alpha^{6k+2}) = (\operatorname{Tr}_n((\alpha^{3k+1})^2)) = (\operatorname{Tr}_n(\alpha^{3k+1}))^2=0$ and $\operatorname{Tr}_n(\alpha^{3k+2})=0 \implies \operatorname{Tr}_n(\alpha^{6k+4}) = (\operatorname{Tr}_n((\alpha^{3k+2})^2)) = (\operatorname{Tr}_n(\alpha^{3k+2}))^2=0$.

\begin{lemma}\label{lem:2-balanced_Gold} Let $\Phi_F:\mathcal{F}_{2,n} \to \Ftn$ be a 2-balanced function, where $n$ is even, whose associated partition of $\Ftn^*$ is given by $(V_1,V_2)$, where  $V_1 = \{x^3 \mid x\in \Ftn^*\}=C_0$, s.t. $\gamma_i = \frac{|\Phi_F^{-1}[V_i]|}{|V_i|}$ for $i\in \{1,2\}$. Then, $S_{v,F}$ is $2$-valued. In particular, $S_{v,F} =  \eta_{i-1}(\gamma_1-\gamma_{2}) - \gamma_2$, when $v\in V_i$, $i=1,2$. Conversely, if $F$ is APN and, for every $v\in \Ftn^*$, $S_{v,F}\in \{\theta_1 ,\theta_2\}$ with $\theta_1\not=\theta_2$ such that either $\{v\in \Ftn : S_{v,F} = \theta_1\} = C_0$ or $\{v\in \Ftn : S_{v,F} = \theta_2\}=C_0$. Then $\Phi_F$ is $2$-balanced. 
\end{lemma}

\begin{proof} By equation (\ref{eq:Sv_2valued}),
$$S_{v,F}=\gamma_1\sum_{c\in C_0}(-1)^{\operatorname{Tr}_n(vc)}+\gamma_2\left(\sum_{c\in C_1}(-1)^{\operatorname{Tr}_n(vc)}+\sum_{c\in C_2}(-1)^{\operatorname{Tr}_n(vc)}\right).$$
We will consider the following two cases.
\begin{itemize}
    \item[1)] Suppose that $v\in V_1=C_0$.  Then, we have $vC_j=C_j$ for each $ j\in\{0,1,2\}.$ Consequently,
$$S_{v,F}=\gamma_1\eta_0+\gamma_2(\eta_1+\eta_2)=\gamma_1\eta_0+\gamma_2(-1-\eta_0) =\eta_0(\gamma_1-\gamma_2)-\gamma_2.$$
    \item[2)] Now suppose that $v\in V_2=C_1\cup C_2$. There are two possibilities:
    \begin{itemize}
        \item[$1^\circ$]  

 If $v\in C_1$, then $vC_0=C_1,$ $vC_1=C_2,$ and $vC_2=C_0$, and therefore
$$S_{v,F}=\gamma_1\eta_1+\gamma_2(\eta_2+\eta_0)=\eta_1(\gamma_1-\gamma_2)-\gamma_2.$$

          \item[$2^\circ$] Similarly, if $v\in C_2$, then $vC_0=C_2$, $vC_1=C_0$ and $vC_2=C_1$, so
$$S_{v,F}=\gamma_1\eta_2+\gamma_2(\eta_0+\eta_1)=\eta_2(\gamma_1-\gamma_2)-\gamma_2.$$
    \end{itemize}
Since $\eta_1=\eta_2$ (as $n$ is even), $S_{v,F}$ takes a single value for all $v\in V_2$.
\end{itemize}
Now, suppose that $S_{v,F} \in \{\theta_1,\theta_2\}$ and, w.l.o.g., suppose that  $V_1 := \{v\in \Ftn^*:S_{v,F}=\theta_1\} = C_0$. Then  $V_2=C_1\cup C_2$. Consider the function $\psi_F:\Ftn\to\mathbb{C}$ given by \eqref{eq:psi}.
Since $\widehat{\psi_F}(0)=|\mathcal{F}_{2,n}|$, we get 
\begin{align}
     \psi_F(c) &= 2^{-n} \left( |\mathcal{F}_{2,n}|+\theta_1\sum_{v\in V_1}(-1)^{\operatorname{Tr}_n(vc)} + \theta_2\sum_{v\in V_2}(-1)^{\operatorname{Tr}_n(vc)}\right)\nonumber \\
     &=  2^{-n} \left( |\mathcal{F}_{2,n}|+\theta_1\sum_{v\in C_0}(-1)^{\operatorname{Tr}_n(vc)} + \theta_2\left(\sum_{v\in C_1}(-1)^{\operatorname{Tr}_n(vc)}+\sum_{v\in C_2}(-1)^{\operatorname{Tr}_n(vc)}\right)\right)\label{eq:psi_gold}.
\end{align}
Consider the following two cases.
\begin{itemize}
    \item[1)] If $c\in V_1=C_0$, then equation (\ref{eq:psi_gold}) becomes
$ \psi_F(c)=2^{-n} \left( |\mathcal{F}_{2,n}|+\theta_1\eta_0 + \theta_2\left(\eta_1+\eta_2\right)\right).$
    \item[2)] If $c\in V_2=C_1\cup C_2$, then there are two possibilities:
    \begin{itemize}
        \item[$1^\circ$]  If $c\in C_1$, then $ \psi_F(c)=2^{-n} \left( |\mathcal{F}_{2,n}|+\theta_1\eta_1 + \theta_2\left(\eta_2+\eta_0\right)\right).$
          \item[$2^\circ$] Similarly, if $c\in C_2$, then 
          $\psi_F(c)=2^{-n} \left( |\mathcal{F}_{2,n}|+\theta_1\eta_2 + \theta_2\left(\eta_0+\eta_1\right)\right).$
    \end{itemize}
\end{itemize}
Thus, it must be that, for $c\not=0$, $$\psi_F(c) = \begin{cases}
    2^{-n} \left( |\mathcal{F}_{2,n}|+\eta_0(\theta_1-\theta_2)-\theta_2\right),& c\in C_0,\\
    2^{-n} \left( |\mathcal{F}_{2,n}|+\eta_1(\theta_1-\theta_2)-\theta_2\right),& c\in C_1\cup C_2.
\end{cases}$$
In other words, $\Phi_F$ is $2$-balanced.\end{proof}

\begin{theo} \label{theo:2bal_nu} Let $F:\Ftn\to \Ftn$ be any function, for $n$ even. The following statements are equivalent.
\begin{enumerate}[i)]
    \item The function $F$ is $2$-CAPN, where one of the partition sets is an $\frac{n}{2}$-dimensional subspace with the zero vector removed (resp. the set of nonzero cubes).
    \item The function $\Phi_F$ is $2$-balanced with one of the parts being an $\frac{n}{2}$-dimensional subspace with the zero vector removed (resp. the set of nonzero cubes).
    \item $F$ is APN and, for each $v\not=0$, it holds $S_{v,F} \in \{\theta_1,\theta_2\}$ for $\theta_1\not=\theta_2 \in \mathbb{Z}$ and one of $\{v \in \Ftn : S_{v,F}= \theta_1\}$ or $\{v \in \Ftn : S_{v,F}= \theta_2\}$ is an $\frac{n}{2}$-dimensional subspace without the zero vector (resp. the set of nonzero cubes).
    \item $F$ is APN and, for each $v\not=0$, the sum-of-square indicator $\nu(F_v)\in \{\Theta_1,\Theta_2\}$ takes on two values and one of $\{v \in \Ftn : \nu(F_v)= \Theta_1\}$ or $\{v \in \Ftn : \nu(F_v)= \Theta_2\}$ is an $\frac{n}{2}$-dimensional subspace without the zero vector (resp. the set of nonzero cubes).\footnote{In \cite{Gillot2026}, the property of a function whose sum-of-square indicator (equivalently, its fourth Walsh moments) takes on $k$ values was called $k$ levels of fourth spectral moments.}
\end{enumerate}
\end{theo}

\begin{proof} The equivalence of $i)$ and $ii)$ is given by Lemma \ref{lemma:characterization_even} (resp. Lemma \ref{lem:2-balanced_Gold}). The equivalence of $ii)$ and $iii)$ follows at once from Lemma \ref{lem:Sv}.
\end{proof}

\begin{ex}
     Dillon's APN permutation $F:\mathbb{F}_{2^6}\to \mathbb{F}_{2^6}$ is $2$-CAPN since  $ \forall x \in V_{1}, \gamma_1 = |\phi^{-1}_F(x)|=160 $ and $\forall x \in V_{2},\gamma_2  = |\phi^{-1}_F(x)|=208 $, where $|V_1| = 56$ and $|V_2|=7$. In this case we have $|\mathcal{F}_{1}| =8960$ and  $|\mathcal{F}_2| =1456$.  Moreover, $$W=V_2\cup\{0\} = \langle 1,\omega+\omega^5,\omega^3+\omega^4 \rangle,$$ where $\omega$ is a primitive element of $\F_{2^6}^*$ such that $\omega^6 + \omega^4 + \omega^3 + \omega + 1=0$. 
     Its orthogonal complement is $W^\perp = \langle w^4, w^5+ w, w^5+ w^2 +1 \rangle $.
     This yields that $S_{v,F}$ is $2$-valued and $$S_{v,F} = \begin{cases}
   -208, & v\not\in W^\perp;\\
   176, &  v \in W^\perp \setminus \{0\},\\
\end{cases}$$
by Lemma \ref{lemma:characterization_even}, that is, the value $208$ is taken $56$ times, whereas the value $176$ is taken $7$ times. Additionally, we have $$ \nu(F_v) = \begin{cases}
   7168 = 2^{2n} + 2^{3s} - 2^{3n-2s} , & v\not\in W^\perp;\\
   16384 = 2^{n+2s}, &  v \in W^\perp\setminus \{0\},
   \end{cases}$$ 
where $s=\frac{n}{2}+1=4.$
\end{ex}

\begin{ex}\label{ex:Gold}
 The Gold function $F:\Ftn \to \Ftn$ given by $F:x\mapsto x^{2^i+1}$, where $\gcd(i,n)=1$, is 2-CAPN for $n$ even. Indeed, when $n$ is even, then there are bent and non-bent plateuaed components, so that $F_v$ is non-bent plateaued if and only if $v$ is a cube in
$\mathbb{F}_{2^n}^*$, with amplitude $\lambda = 2^{\frac{n}{2}+1}$ and $F_v$ is bent otherwise. Since
$\mathbb{F}_{2^n}^*$ has $\frac{2^n-1}{3}$ cubes, there are $\frac{2^n-1}{3}$ non-bent plateaued components and $\frac{2(2^n-1)}{3}$ bent components. Hence 
\begin{equation}\label{eq:nu_Gold}
    \nu(F_v) = \begin{cases}
    2^{2n+2}, & v=w^3 \text{ for some } w\in \Ftn^*\\
    2^{2n}, & \text{otherwise.}
\end{cases}
\end{equation}
Theorem \ref{theo:2bal_nu}, implies that $\Phi_F$ is 2-balanced.  In this case, $$S_{v,F} = \begin{cases}
   \frac{2^{n+1}-2^{2n+1}}{24}, & v=w^3 \text{ for some } w\in \Ftn^*,\\
    \frac{2^{2n}+2^{2n+1}}{24}, & \text{otherwise.}
\end{cases}$$

\end{ex}

Example \ref{ex:Gold} shows that a $2$-CAPN function is not necessarily an (APN) permutation. On the other hand, all known scarce examples of APN permutations are either CAPN, for $n$ odd, or, $2$-CAPN, for $n$ even. There may be a structural/combinatorial property on the set of $2$-flats which is preserved by $\Phi_F$, when $F$ is a permutation. Hence, it is natural to ask whether an APN permutation is forced to be either CAPN or $2$-CAPN. 

\begin{op} For an APN permutation $F$, is it true that $\Phi_F$ is balanced, for $n$ odd, and $2$-balanced, for $n$ even?
\end{op}

\section{Conclusions}\label{sec:conc} 
In this article, we have provided several results concerning notions related to APN mappings. In particular, we have extended the study of the breaking porpery to general mappings and characterized the $2$-breaking property of APN functions. Namely, we specifed exactly the number of unbroken 2-dimensional flats for APN functions. Moreover, using this property, we derived a lower bound on the vectorial nonlinearity of $(n,n)$-functions. We provided an in-depth analysis of the relations among the breaking property, and two recently introduced generalizations of the APN property, strongly non-normality and sum-freedom. We also introduced the concept of {\em strong breaking} of $k$-dimensional flats, as a refinement of the breaking property, and derived several structural results for both notions. Moreover, a characterization of a subclass of APN functions in terms of the $2$-strongly breaking property has been given. Finally, we refined a different perspective of the non-vanishing property via a natural character transformation first employed by Carlet. In particular, we derived a precise value for the total sum of the sum-of-square indicators of a polynomial $F$. Notably, we solved Open Problem 4 in \cite{BeCaChLa06} by employing these tools. With the help of this framewrok, we introduced some balancedness properties related to polynomials. One of them is equivalent to the notion of \emph{component-wise APNess}, for odd dimensions. Then, we have proposed a natural extension to any dimension via the property of being balanced-like. A final balancedness property of APN functions has been examined, which we show is satisfied by notable functions such as Dillon's APN permutation and the Gold functions. We suggest four interesting research challenges: show that any APN function $F$ is $k$-breaking for $n-1\geq k\geq \lfloor \frac{n}{2} \rfloor$+1, $n\geq7$; provide a tight lower bound on the vectorial nonlinearity of quadratic APN functions; prove that a $k$-th-sum order-free polynomial that is also $(k-1)$-st-sum order-free is necessarily $k$-breaking, for $k>3$; and the challenging task of showing that $\Phi_F$ is balanced, for $n$ odd, and $2$-balanced, for $n$ even, whenever $F$ is an APN permutation $F$.

\addresseshere

\newpage
\section{Appendix: computational data}\label{sec:comput}

Table \ref{table: APNs of deg le 3} summarizes the known APN functions for $n=6$, up to CCZ equivalence, with an additional row for the known APN function. Table \ref{table: APNs unbroken flats} shows the number of unbroken $2$-flats, the number of non-strongly broken $2$-flats, and the preimage distribution of $\Phi_F$, for functions in Table \ref{table: APNs of deg le 3}. Similarly, Table \ref{table:unbroken-vanishing-3} displays the number of unbroken $k$-flats and the number of non-strongly broken $k$-flats for $k\in\{3,4\}.$

\begin{table}[H]
    \captionsetup{font=scriptsize}
	\caption{\scriptsize Representatives of CCZ-equivalence classes of APN functions in six variables. Here, $\alpha$ is a root of the primitive polynomial $X^6+X^4+X^3+X+1\in\F_2[x]$.}\label{table: APNs of deg le 3}
	\centering
	\resizebox{0.75\textwidth}{!}{
		\begin{tabular}
			{|c|l|c|} \hline $F$ & \multicolumn{1}{c|}{Univariate representation of $F$} & Preimage dist. \\ \hline
			$D_{1\phantom{3}}$ & $x^{3}$ & $\{* 1, 3^{21} *\}$ \\
			$D_{2\phantom{3}}$ & $x^{3} +\alpha^{11}x^{6} + \alpha x^{9}$ & $\{* 1, 3^{21} *\}$\\
			$D_{3\phantom{3}}$ & $\alpha x^{5} + x^{9} + \alpha^{4}x^{17} + \alpha x^{18} + \alpha^{4}x^{20} + \alpha x^{24} + \alpha^{4}x^{34} + \alpha x^{40}$  & $ \{* 1^{25}, 2^{10}, 3^{5}, 4 *\}$ \\
			$D_{4\phantom{3}}$ & $\alpha^{7}x^{3} + x^{5} + \alpha^{3}x^{9} + \alpha^{4}x^{10} + x^{17} + \alpha^{6}x^{18}$& $\{* 1^{26}, 2^{11}, 3^{4}, 4 *\}$ \\
			$D_{5\phantom{3}}$ & $x^{3} + \alpha x^{24} + x^{10}$ & $\{* 1^{15}, 2^{14}, 3^{7} *\}$ \\
			$D_{6\phantom{3}}$ & $x^{3} + \alpha^{17}(x^{17} + x^{18} + x^{20} + x^{24})$ & $\{* 1^{25}, 2^{9}, 3^{4}, 4, 5 *\}$\\
			$D_{7\phantom{3}}$ & $x^{3} + \alpha^{11}x^{5} + \alpha^{13}x^{9} + x^{17} + \alpha^{11}x^{33} + x^{48}$ & $\{* 1^{25}, 2^{13}, 4^{2}, 5 *\}$  \\
			$D_{8\phantom{3}}$ & $\alpha^{25}x^{5} + x^{9} + \alpha^{38}x^{12} + \alpha^{25}x^{18} + \alpha^{25}x^{36}$ & $\{* 1^{22}, 2^{15}, 3^{4} *\}$  \\
			$D_{9\phantom{3}}$ & $\alpha^{40}x^{5} + \alpha^{10}x^{6} + \alpha^{62}x^{20} + \alpha^{35}x^{33} + \alpha^{15}x^{34} + \alpha^{29}x^{48}$ & $\{* 1^{23}, 2^{10}, 3^{3}, 4^{3} *\}$ \\
			$D_{10}$ & $\alpha^{34}x^{6} + \alpha^{52}x^{9} + \alpha^{48}x^{12} + \alpha^{6}x^{20} + \alpha^{9}x^{33} + \alpha^{23}x^{34} + \alpha^{25}x^{40}$ & $\{* 1^{26}, 2^{11}, 3^{4}, 4 *\}$\\
			$D_{11}$ & $x^{9} + \alpha^{4}(x^{10} + x^{18})+ \alpha^{9}(x^{12} + x^{20} + x^{40})$  & $\{* 1^{28}, 2^{13}, 3^{2}, 4 *\}$ \\
			$D_{12}$ & $\alpha^{52}x^{3} + \alpha^{47}x^{5} + \alpha x^{6} + \alpha^{9}x^{9} + \alpha^{44}x^{12} + \alpha^{47}x^{33} + \alpha^{10}x^{34} + \alpha^{33}x^{40}$  & $\{* 1^{27}, 2^{12}, 3^{3}, 4 *\}$ \\
			$D_{13}$ & $\alpha(x^{6} + x^{10} + x^{24} + x^{33}) + x^{9} + \alpha^{4}x^{17}$   & $\{* 1^{26}, 2^{11}, 3^{4}, 4 *\}$   \\  \hline
			APN Perm. & long univariate representation  & $\{* 1^{64} *\}$\\ \hline
            
			\multirow{2}{*}{$EP$} & $x^3 + \alpha^{17}(x^{17} + x^{18} + x^{20} + x^{24}) + \alpha^{14}(\alpha^{18}x^9 + \alpha^{36} x^{18} + \alpha^9 x^{36} $ &  \multirow{2}{*}{$\{* 1^{21}, 2^{15}, 3, 4, 6 *\}$}  \\ 
			& $+ x^{21} + x^{42}+\operatorname{Tr}_n(\alpha^{27} x + \alpha^{52} x^3 + \alpha^6 x^5 + \alpha^{19} x^7 + \alpha^{28} x^{11} + \alpha^2 x^{13}))$ & \\ \hline		
		\end{tabular}
	}
\end{table}

\begin{table}[H]
    \captionsetup{font=scriptsize}
	\caption{\scriptsize Number of unbroken and non-strongly broken 2-flats of some APN functions in $6$ variables.}\label{table: APNs unbroken flats} 
	\centering
	\resizebox{0.55\textwidth}{!}{
		\begin{tabular}{|c|c|c|c|} \hline $F$ &  $|\mathcal{U}_{6,F}|$ & $|\mathcal{NS}_{6,F}|$ & Preimage dist. of $\Phi_F$ \\ \hline
			$D_{1}$ & $21$ & 0 & $\{* 144^{42}, 208^{21} *\}$\\\ 
			$D_{2\phantom{3}}$ & $21$ & 0 & $\{* 144^{42}, 208^{21} *\}$\\\
			$D_{3\phantom{3}}$ & $ 9$ & 13 & $\{* 80,  112^{5}, 144^{20}, 176^{26}, 208^{10}, 240\}$ \\
			$D_{4\phantom{3}}$ & $8$ & 13 & $\{* 144^{42}, 208^{21} *\}$  \\
			$D_{5\phantom{3}}$ & $7$ & 21 & $\{* 144^{42}, 208^{21} *\}$ \\
			$D_{6\phantom{3}}$ & $18$& 9 & $\{* 80,  112^{5}, 144^{20}, 176^{26}, 208^{10}, 240\}$ \\
			$D_{7\phantom{3}}$ & $18$& 15 & $\{* 80^{6}, 112^{10}, 144^{10}, 176^{16}, 208^{15}*, 240^6\}$ \\
			$D_{8\phantom{3}}$ & $4$ & 9 & $\{* 80,  112^{5}, 144^{20}, 176^{26}, 208^{10}, 240\}$\\
			$D_{9\phantom{3}}$ & $15$& 17 & $\{* 80,  112^{5}, 144^{20}, 176^{26}, 208^{10}, 240\}$\\
			$D_{10}$ & $8$  & 13 & $\{* 80,  112^{5}, 144^{20}, 176^{26}, 208^{10}, 240\}$  \\
			$D_{11}$ & $6$  & 13 & $\{* 112^6, 144^{22}, 176^{24}, 208^9, 240^2*\}$\\
			$D_{12}$ & $7$  & 14 & $\{* 112^6, 144^{22}, 176^{24}, 208^9, 240^2*\}$ \\
			$D_{13}$ & $8$  & 11 & $\{* 80,  112^{5}, 144^{20}, 176^{26}, 208^{10}, 240\}$\\  \hline
			APN Perm & $0$ & 0 & $\{* 160^{56}, 208^7*\}$ \\ \hline
			$EP$ & $25$ & 21 & $\{* 112, 128^6, 144^3,160^{32},176^7,192^{10},208^4*\}$\\ \hline		
		\end{tabular}
	}
\end{table}

\begin{table}[H]
	  \captionsetup{font=scriptsize}
	\caption{Number of unbroken and non-strongly broken $3$-flats and $4$-flats of APN functions in $6$ variables.}\label{table:unbroken-vanishing-3}
	\centering
	\resizebox{0.45\textwidth}{!}{
		\begin{tabular}
			{|c|c|c|c|c|} \hline
			$F$ & $|\mathcal{UB}_{6,3,F}|$ & $|\mathcal{NSB}_{6,3,F}|$ & $|\mathcal{UB}_{6,4,F}|$ & $|\mathcal{NSB}_{6,3,F}|$\\
            \hline
			$D_{1\phantom{3}}$ & 9 & 3285 & 0 & 714\\ 
			$D_{2\phantom{3}}$ & 0 & 3285 & 0 & 714  \\
			$D_{3\phantom{3}}$ & 0 & 7779   & 0 & 546  \\
			
			$D_{4\phantom{3}}$ & 1 & 7721 & 0 & 549   \\
	
			$D_{5\phantom{3}}$ &23 & 6561 & 0 & 203  \\
	
			$D_{6\phantom{3}}$ & 1 & 7332  & 0 & 497 \\
	
			$D_{7\phantom{3}}$ & 1 & 7073 & 0 & 454 \\
	
			$D_{8\phantom{3}}$ & 0 & 7617  & 0 & 473 \\
		
			$D_{9\phantom{3}}$ & 10 & 6640  & 0 & 329 \\
			
			$D_{10\phantom{3}}$ & 1 & 7711 & 0 & 531  \\
			
			$D_{11\phantom{3}}$ & 1 & 7809 & 0 & 559  \\
			
			$D_{12\phantom{3}}$ & 1 & 7586& 0 & 494  \\
			
			$D_{13\phantom{3}}$ & 0 & 7502 & 0 & 499 \\
			\hline
			APN Perm & 0 & 184 & 0 & 812 \\\hline
			EP & 0 & 4336 & 0 & 391 \\\hline
			
		\end{tabular}
	}
\end{table}

\end{document}